\documentclass[10pt,journal,compsoc]{IEEEtran}
\usepackage{graphicx}
\usepackage{verbatim}
\usepackage{subfig}
\usepackage{cite}
\usepackage{url}
\usepackage{algorithm}
\usepackage[noend]{algorithmic}
\usepackage{amsmath}
\usepackage{mathtools}
\usepackage{geometry}
\usepackage{caption}
\usepackage{comment}
\usepackage{fixltx2e}

\usepackage{amsthm}
\usepackage[utf8]{inputenc}
\usepackage[english]{babel}
\newtheorem{theorem}{Theorem}

\DeclareRobustCommand*{\IEEEauthorrefmark}[1]{%
  \raisebox{0pt}[0pt][0pt]{\textsuperscript{\footnotesize\ensuremath{#1}}}}

\begin{document}

\date{}


\title{Enhancing Cloud Storage with Shareable Instances for Social Computing}

\author{

\IEEEauthorblockN{
Ying Mao\IEEEauthorrefmark{1},
Peizhao Hu\IEEEauthorrefmark{2},Member, IEEE}
\IEEEcompsocitemizethanks{\IEEEcompsocthanksitem Ying Mao\IEEEauthorrefmark{1} is with the Department of Computer and Information Science, Fordham University, New York, 10458, Email: ymao41@fordham.edu
Email: maoy@tcnj.edu
\IEEEcompsocthanksitem Peizhao Hu\IEEEauthorrefmark{2} is with Department of Computer Science, Rochester Institute of Technology, Rochester, NY,14623, Email: ph@cs.rit.edu
}
}

\IEEEtitleabstractindextext{
\begin{abstract}
Cloud storage plays an important role in social computing. This paper aims to develop a cloud storage management system for mobile
devices to support an extended set of file operations. 
Because of the limit of storage, bandwidth, power consumption and other resource restrictions, most
existing cloud storage apps for smartphones do not keep local copies of
files. This efficient design, however, limits the application capacities.
In this paper, we attempt to extend the available file operations for
cloud storage service to better serve smartphone users. We develop an efficient and 
secure file management system, Skyfiles, to support
more advanced file operations. The basic idea of our design  is to utilize cloud
instances to assist file operations. Particularly, Skyfiles supports
downloading, compressing, encrypting, and converting operations, as well as file transfer between
two smartphone users' cloud storage spaces. In addition, we propose a
protocol for users to share their idle instances. All file operations
supported by Skyfiles can be efficiently and securely accomplished with
either a self-created instance or shared instance. 
\end{abstract}
}

\maketitle

\section{Introduction}
With recent advances, smartphones have become the most revolutionary devices nowadays. According to a Pew's survey made in 2015, about $68\%$ of U.S. consumers own a smartphone. Consequently, a variety of applications
have been developed to meet users' all kinds of demands (more than 700,000 apps
in both Apple Store and Google Play Store). With the recent development of cloud computing~\cite{chen2020woa, mao2020resource, mao2017draps},
edge computing~\cite{acharya2019edge, acharya2019workload, harvey2017edos} and deep learning~\cite{fu2019progress, zheng2019flowcon, zheng2019target}, today's smartphones have gone far beyond a mobile telephone as they have seamlessly dissolved in people's daily life in every perspective. Among various applications, social networking apps play an important role. 
The concepts behind ``social networking'' aren't anything new, we have been looking for ways to connect, network, and promote with each other ever since there have been humans. In the age of smartphones, where we used to have handshakes, word-of-mouth referrals, and stamped letters, today’s relationships are often begun and developed on Twitter~\cite{Twitter}, Google+\cite{gg}, and Facebook~\cite{fb}. The spirit of mobile social networks is data sharing.
However, the mobile devices have encountered specific challenges due to the limitations on these devices. 
First, the storage capacity of smartphones or Tablets is limited compared to
regular desktops and laptops. Second, the network bandwidth of the cellular network is limited. At this point, major U.S. mobile networks carriers rarely provide unlimited data plans with full speed and the service scalability is limited by fundamental constraints. Finally, energy consumption is a critical issue for these devices.

In this paper, we consider the applications of the cloud
storage service that is a major approach to extend the storage space for smartphones. Representative products
include iCloud~\cite{icloud}, Dropbox~\cite{dropbox}, Box.com~\cite{box}, Google Drive~\cite{googledrive},
and others~\cite{skydrive,ubuntuone}. Basically, each user holds a certain remote storage space in cloud and
can access the files from different devices through the Internet. Synchronization and file consistence are guaranteed
in these cloud storage services. 
To address the above limitations on these devices, the most existing applications
for cloud storage service follow one important design principle---not keeping local copies of the files
stored in cloud. Because smartphones or IoTs may not have sufficient space to hold all the files,
and because downloading those files consumes a lot of bandwidth and battery power,
only meta data is kept on the devices by default, instead. Though this design is efficient,
it limits the capabilities of the apps. Some file operations that can be easily done with local
copies become extremely hard for
smartphone users, e.g., compressing files and transferring files to another user.

In this paper, we develop Skyfiles~(extended \cite{mao2014skyfiles}), a system for smartphone and IoT users to manage their
files in cloud storage. Our basic idea is to launch a cloud instance to assist users
to accomplish the file operations. It is motivated by the fact that the cloud instance
is inexpensive (sometimes free). For example, Amazon Web Service (AWS)~\cite{aws} provides 750
free Micro instance hours per month. By delegating tasks to a cloud instance, users can
address the above constraints in terms of storage space, bandwidth usage, and energy consumption for
file operations. Our Skyfiles system does not store local copies of files on smartphones, but possesses the following new features:

\begin{itemize}
\item
It extends the available file operations for mobile devices to a
more enriched set of operations including downloading, compressing, encrypting, and converting operations;
\item
It includes a protocol for two smartphone users to transfer files from one's cloud storage space to the other's cloud storage;
\item
It includes a secure solution for all of above operations to use shared cloud instances, i.e., shared instances created by other users.
\end{itemize}

The rest of this paper is organized as follows. Section~\ref{sec:rel}
overviews the related work and Section~\ref{sec:back} introduces background
information about cloud storage service. In Section~\ref{sec:sol},
we present the basic architecture of Skyfiles.
Section~\ref{sec:sol2} includes detailed design of the file operations in Skyfiles.
We evaluate the performance of Skyfiles in Section~\ref{sec:eval} and conclude in Section~\ref{sec:conclude}.


\section{Related Work}
\label{sec:rel}

In the past few years, cloud computing and storage technology gain more and more attention in mobile applications.
In~\cite{Survey}, a survey studies offloading
computation to the cloud side with the objective of extending smartphone battery life.
MAUI\cite{MAUI}, Cuckoo\cite{mobicase_2010} and ThinkAir\cite{ThinkAir} implement an Android framework on 
top of the existing runtime system. These three systems are easy
to deploy because they only need to access to the program source codes, and they do not require modifying the existing operating system. 
They provide a dynamic runtime system, which can, at runtime, decide whether a part of
an application is better to be executed locally or remotely.
Besides offloading computation for the mobile devices, 
SmartDiet\cite{SmartDiet} aims at offloading communication-related tasks to cloud in order to save energy of smartphones.
Authors in~\cite{mao2016mobile} study a message sharing scheme without infrastructure support.

In addition to offloading services, another major utilization of cloud computing for mobile users is to enrich the storage space on smartphones.
CloudCmp~\cite{CloudCmp} studies on a comparison about user achieved performance among different storage
providers by running benchmarks.
In~\cite{IWQoS_2012}, the authors focus on the impact of virtualization on Dropbox-like cloud storage systems. 
Focusing on Amazon Web Service, a cloud storage infrastructure provider, ~\cite{aws-study} shows that the perceived performance at the client 
is primarily dependent on the client's network capabilities and the network performance between client and AWS data center.
Furthermore, the authors of~\cite{Inside_dropbox} investigate Dropbox users to
understand characteristics of personal cloud storage services on mobile device. Their results show possible
performance bottlenecks caused by both the current system architecture and the storage protocol.

While cross-platform cloud storage services provide great usability, the increasing
amount of personal or confidential data entrusted to these
services poses a growing risk to privacy~\cite{s-p,proof,Wizards,control}.
There are some prior works attempting to use cloud computing to enhance security and privacy of mobile devices.
Authors in~\cite{privacy-survey} provide a survey to explore the security and privacy issues in cloud computing.
CloudShield\cite{CloudShield} presents
an efficient anti-malware smartphone patching with a P2P network on the cloud. CloudShield can stop worm spreading 
between smartphones by using cloud computing.
Clone2Clone~\cite{Clone2Clone} uses device-clones
hosted in various virtualization environments in both private
and public cloud to boost mobile performance and strengthen its security. 
The authors of \cite{caas} propose the Confidentiality as a Service (CaaS) paradigm to provide usable confidentiality
and integrity for the bulk of users who think the current security mechanisms are too complex or require too much effort for them.
In their paradigm, CaaS separates capabilities, requires less trust from cloud or CaaS provider, leverages
existing infrastructure, and performs automatic key management. They test CaaS on Facebook, Dropbox and some popular service providers.


\section{Background Of Cloud Storage}
\label{sec:back}

This section introduces the background information about cloud storage service. 
In this chapter, most of our experiments are conducted on Dropbox platform as a representative
service provider. Accordingly, we briefly describe the Dropbox architecture and functions for user applications in this section.

As a leading solution in personal cloud storage market, 
Dropbox provides
cross-platform service based on Amazon Simple Storage Service(S3) for both
desktop and mobile users.

\begin{figure}[ht]
\centering
\includegraphics[width=0.36\textwidth]{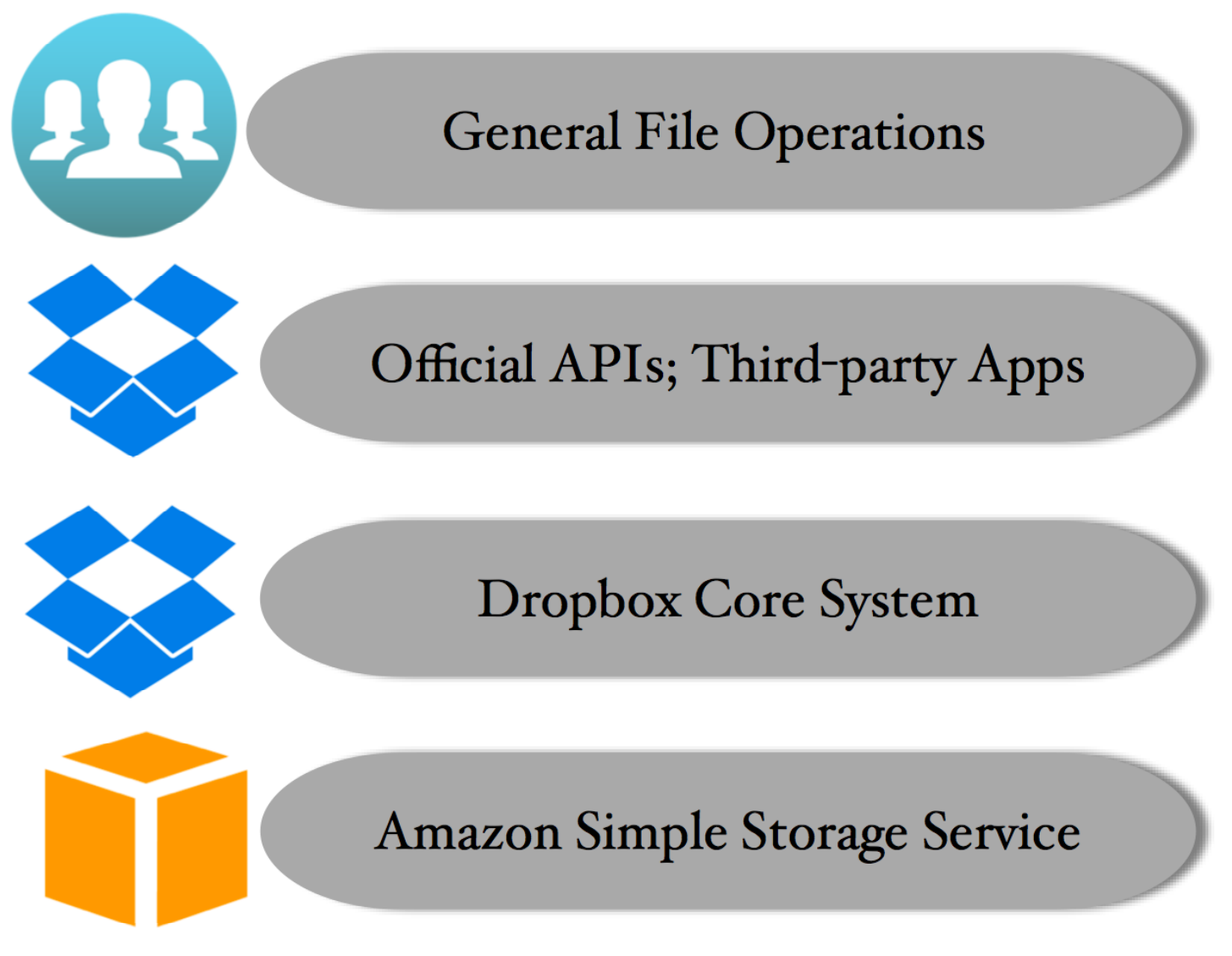}
\caption{Architecture of Dropbox services}
\label{fig:dropbox}
\end{figure}
%
The architecture of Dropbox follows a layered structure (Fig.~\ref{fig:dropbox}). At
bottom level, Amazon S3 infrastructure provides a basic
interface for storing and retrieving data. The above layer is Dropbox core system which interacts with S3 storage service and
serves higher level applications. The top layer is the official Dropbox application and a set of APIs for developers to build
third party applications. 
To manage third party applications, Dropbox assigns each of
them an unique {\em app key} and {\em app secret} which can be used to identify one particular application.
When a client launches an application, Dropbox server follows the OAuth
\cite{oauth} for authorization. 
\begin{figure}[ht]
\centering
\includegraphics[width=0.49\textwidth]{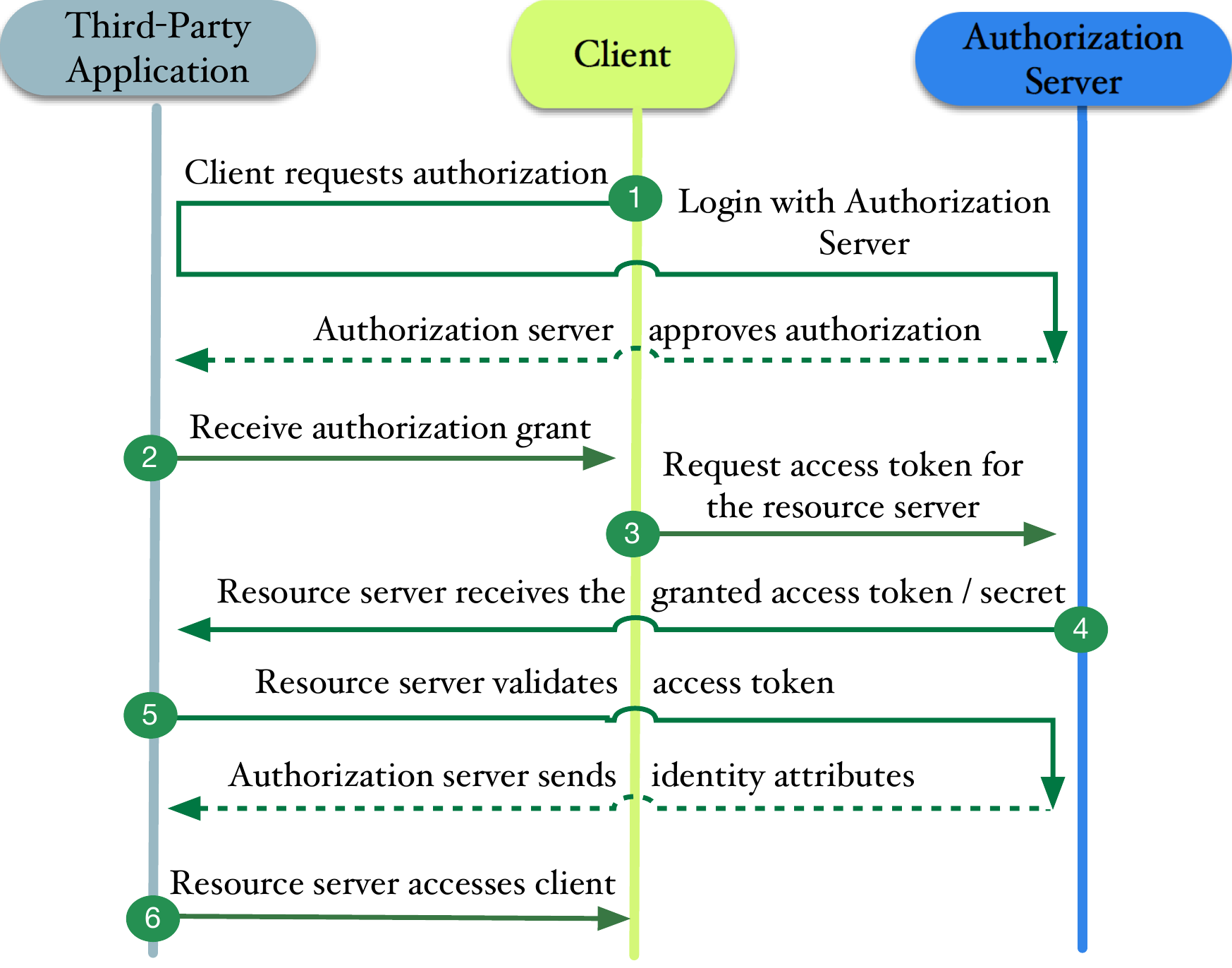}
\caption{OAuth workflow}
\label{fig:oauth}
\vspace{-0.1in}
\end{figure}

Fig.~\ref{fig:oauth} shows the
basic authorization flow. 
When launched by a client, the third party app 
contacts the authorization server (Dropbox server) to obtain an one time request token and request secret. 
Then, the app uses them to form a redirect link and present
the link to the client. 
When accessing this link, the client will be
prompted to login his Dropbox account and the Dropbox server will verify the
redirect link and the client's login information (Step 1 to 3).
After a successful login,
the server will return an {\em access token} and an {\em access secret} to the
application, 
which grants access permissions on the client's data
stored on Dropbox (Step 4).
Then, resource server sends the token directly to the authorization server to validate it (Step 5).
Finally, after the validation, third party applications can access the client's cloud storage (Step 6). 

\section{Architecture Of Skyfiles}
\label{sec:sol}

In this section, we present the basic architecture of Skyfiles. We first define the file operations that supported by Skyfiles, especially the new extended set of operations. Then we discuss the options of cloud instances we consider in Skyfiles. Finally we describe the major components in our design.

\subsection{Supported File Operations}
\label{sec:fileop}
In this subsection, we first present the file operations that Skyfiles supports. Our solution enables a user to launch an extended set of file operations, called {\em cloud-assisted operations}, on his files. In addition, we support file transfer between two users' cloud storage spaces.

\subsubsection{Single User File Operations}

Skyfiles supports two categories of file operations on a user's own files. The first category is
the {\em basic file operations} that are commonly available in service providers'
APIs such as creating/deleting/renaming a file. The second category is a new
set of advanced file operations that cannot be accomplished without a local copy
of the files. Skyfiles seeks the assistance from a cloud instance to conduct these operations. Thus we call the second category of operations {\em cloud-assisted file operations}.
In particular, we have considered the following four
cloud-assisted operations in Skyfiles:

\begin{itemize}

\item {\bf Download:} This operation allows a user to download files directly to his cloud storage.
Given the location of the target files such as URLs, the conventional way of downloading is to
first obtain the files on user devices and then synchronize with/upload to the cloud storage.
In Skyfiles, the cloud instance will fetch the files and then upload them to the user's cloud
storage, so that the downloading and uploading will not consume mobile device's bandwidth.

\item {\bf Compress:} This operation enables a user to compress existing files or directories
stored in cloud. If user devices hold local copies of the target files, the
operation can be easily accomplished and the generated compressed file can
be uploaded to the cloud storage. In Skyfiles or other similar apps for
mobile devices, however, the actual file contents are not available. Thus we
design an interface where the user can select the target files based on the
local shadow file system with meta data and then the compressing operation
is forwarded to a cloud instance for execution. The instance will fetch the
specified files from the cloud storage, compress them, and upload the
compressed file back to the cloud storage.

\item {\bf Encrypt:} This operation is similar to compression and does not exist in
current apps that do not keep local file copies. In Skyfiles, the user can
choose the target files and the cipher suite including the cryptographic
algorithm and key. The encryption operation will be sent to a cloud instance.
Similarly, the cloud instance will download the target files from the user's
cloud storage, encrypt them, and send the ciphertext back to the cloud storage.

\item {\bf Convert:} The last operation is particularly for media files such as
pictures and video clips. When a user wants to view a picture stored in cloud,
he has to download it to his smartphone. Nowadays, high-resolution picture
files could be very large, but a smartphone user may not benefit from it because
of the limited screen size. In Skyfiles, therefore, a user can specify an
acceptable resolution when view a picture. The request will be processed
by a cloud instance. The original picture will be downloaded to the instance
and then converted to a smaller file according to the user-specified resolution.
Finally, the converted picture is sent to the user.

\end{itemize}

Overall, we develop the above set of advanced file operations for mobile
devices which are impossible to achieve in conventional apps without local
file copies. In Skyfiles, a cloud instance is launched to assist a user to
accomplish the advanced file operations. With the design and framework of Skyfiles, it is 
easy to enrich the supported file operations.

\subsubsection{File Transfer Between Users}

Another important file operation in Skyfiles is the
file transfer between users' cloud storage spaces. While most cloud storage services allow a user
to share files with another user, copying files across different user spaces
is not supported. However, file sharing between users cannot substitute file
transfer (make a copy). With file sharing, a user's actions on the shared
files will affect other users. For example, if a user deletes the shared
files, all the other users lose those files as well. In this subsection, we
consider the file transfer between user spaces illustrated in
Fig.\ref{fig:demo}. Assume that two users carrying smartphones meet with each
other and both have storage spaces on cloud. One user (as the sender) wants
to transfer files in his cloud storage to the other user's (as the receiver)
cloud storage.

In the conventional solution, the sender can download the target files to
his smartphone and send them to the receiver's phone through Internet or
short range connection such as Bluetooth and NFC. Upon receiving the files,
the receiver's phone can upload or synchronize them with the cloud storage.
This solution, however, is not efficient in terms of bandwidth consumption,
especially when the target files are large, e.g., a bunch of pictures or
video clips, as the sender has to download the entire files and the receiver
has to upload all the files. In Skyfiles, we solve the problem by following the same design principal of
utilizing a cloud instance to assist users to transfer files between their
cloud storage spaces.


\subsection{Cloud Instances}

For cloud-assisted file operations, we consider two options of cloud instances to carry out the operations. The first option is a {\em private cloud instance}, which is created by the user. This requires the smartphone user to have an account with a cloud computing service provider such as Amazon AWS~\cite{aws} or Azure~\cite{azure}. The private cloud instance is completely trusted and can be customized with additional packages to support the cloud-assisted file operations. The smartphone user has the root privilege and full control on the private cloud instance.

The second option is to use a {\em shared cloud instance}, i.e., the cloud instance initiated by other users. This option could help address the following two issues of the private cloud instance.
First, launching a cloud instance on-demand for the cloud-assisted file operations incurs a significant
overhead, e.g., the overhead ranges from 15 seconds to 30 seconds in our experiments in
Section~\ref{sec:eval}. In reality, such a long delay is not
suitable for some file operations which require instant
response and may negatively affect users experience. The second
issue is the cost of launching cloud instances. Although cloud
service is inexpensive, frequently starting cloud instances may
still increase the cost of users. High cost could be caused by
users' misconfiguration and incorrect design of apps. Users can
certainly monitor the usage and charge on their cloud service
account to avoid unexpected costs. But it could greatly limit
the features of Skyfiles. In Skyfiles, users are allowed to share cloud instances with each other. It is motivated by the fact that cloud service providers charge the instance service at a certain time granularity. For example, AWS, Microsoft Azure~\cite{azure} , and HP Cloud~\cite{hpcloud}
charge the usage of cloud instance at the granularity of an
hour. For regular file operations, it is an excessive time
period. If a user starts a cloud instance for a file operation,
he does not have to terminate the instance when the operation
is done. The instance can be kept running to serve other users
until an additional cost is about to be charged.
For example, assume a service
provider charges the instance service in the time unit of an
hour, when a user starts an instance and finishes his file
operations in the first 5 minutes, the instance he started can
stay active for another 55 minutes without extra cost for him.
During this idle time period, if the instance can serve other
users or other file operations from the same user, the overhead
and cost will be both reduced. In addition, there could be long-term instances available from voluntary users and dedicated Skyfiles servers.

\subsection{Framework Of Skyfiles}

In Skyfiles, each mobile device is associated with a cloud storage account.
Similar to other related apps, Skyfiles by default does not keep local
copies of the files stored in cloud because of the storage limit,
bandwidth consumption and battery life. Instead, Skyfiles maintains a {\em shadow file
system} on the mobile device, which includes the meta information of the
files stored in cloud. This local file system is built on service provider's
APIs and synchronized with the cloud storage. 

In Skyfiles, we define a set of {\em file operation instructions} (FOIs) which is the building block of each cloud-assisted file operations. The following Table~\ref{table:foi} lists the FOIs we consider in Skyfiles.
\begin{table}
\centering
\begin{tabular}{|l|l|}
\hline
{\em download(file)}& Download a file to the local host\\
\hline
{\em get(file)}& Get a file from the cloud storage space\\
\hline
{\em put(file)}& Upload a file to the cloud storage space\\
\hline
{\em op(file)}& Conduct a file operation(op) on a local file\\
\hline
{\em push(file)}& Send a file to the smartphone\\
\hline
\end{tabular}
\caption{File Operation Instructions (FOIs)}
\label{table:foi}
\end{table}
Each cloud-assisted file operation defined earlier can be interpreted as a sequence of FOIs. 
Here lists the conversion of the four cloud-assisted file operations we consider in Skyfiles.
\begin{itemize}
\item Download: {\em download(f)} $\rightarrow$ {\em put(f)}

\item Compress: {\em get(f)} $\rightarrow$ {\em f1=op(f)} $\rightarrow$ {\em put(f1)}, where `op' is a compression operation.
    
\item Encrypt: {\em get(f)} $\rightarrow$ {\em f1=op(f)} $\rightarrow$ {\em put(f1)}, where `op' is an encryption operation.
    
\item Convert: {\em get(f)} $\rightarrow$ {\em f1=op(f)} $\rightarrow$ {\em push(f1)}, where `op' is a conversion operation.
\end{itemize}
Note that Skyfiles is a generally defined middleware, and the cloud-assisted file operations and the FOIs can be easily extended to support more file operations.

There are three major components in the architecture of Skyfiles: Skyfiles Agent, Skyfiles Service Program, and Skyfiles Server. 

\begin{itemize}

\item {\bf Skyfiles Agent} (\texttt{SA}) is software module running on a user's smartphone. It interacts with the user's requests, and configures the cloud storage and cloud instance if needed. 
First, it utilizes the standard APIs offered by
the cloud storage service provider, such as Dropbox, to obtain the permission of accessing the
clients' files on the cloud storage infrastructure. Then, it creates a shadow file
system which only includes the meta information on the smartphones. When a user operates
on the files, Skyfiles agent decides where to execute the operation, i.e., either locally on the phone or remotely on a cloud instance.
The decision is made based on the operation type,
the standard API support, and bandwidth/overhead cost. 
In this chapter, we simply consider the two categories of the file operations mentioned earlier in Section~\ref{sec:fileop}, {\em basic file operations} and {cloud-assisted file operations}.
All the basic file operations is processed on the phone by interacting with the cloud storage server via standard APIs. Skyfiles agent
will pass all the cloud-assisted operations to a cloud instance for execution. In practice, however,
the smartphone may also download a copy of the target file and complete the file operations defined in cloud-assisted file operations depending on the bandwidth/energy cost and the cost of launching a cloud instance, e.g., if the target file is small, it might be worthwhile to download it to the phone for the operation rather than launching a cloud instance to complete it. To simplify the problem,
our solution always forward the cloud-assisted operations to a cloud instance although Skyfiles agent
can be generally defined to decide the best strategy for conducting each file operation.

\item {\bf Skyfiles Service Program} (\texttt{SSP}) is a background service running on a participating cloud instance. It receives and processes the requests from a user's smartphone for cloud-assisted file operations. \texttt{SSP} includes the implementation of all the FOIs for the host environment. The communication between a user's smartphone and the \texttt{SSP} is protected by secure sockets such as TLS/SSL. During the execution of the file operation, \texttt{SSP} also sends heartbeat messages to the user's smartphone to report the progress or any exceptions it encounters. 

\item {\bf Skyfiles Server} (\texttt{Serv}) is a centralized trusted server dedicated to Skyiles application. It maintains two major functions. First, the Skyfiles server dispatchs the Skyfiles Service Program to each participating cloud instance, either a private cloud instance or a shared cloud instance. Second, the server keeps a pool of the available shared cloud instances and when receiving a user's request for a shared instance, it will allocate one for the user.

\end{itemize}

\begin{figure}[ht]
\centering
\includegraphics[width=0.48\textwidth]{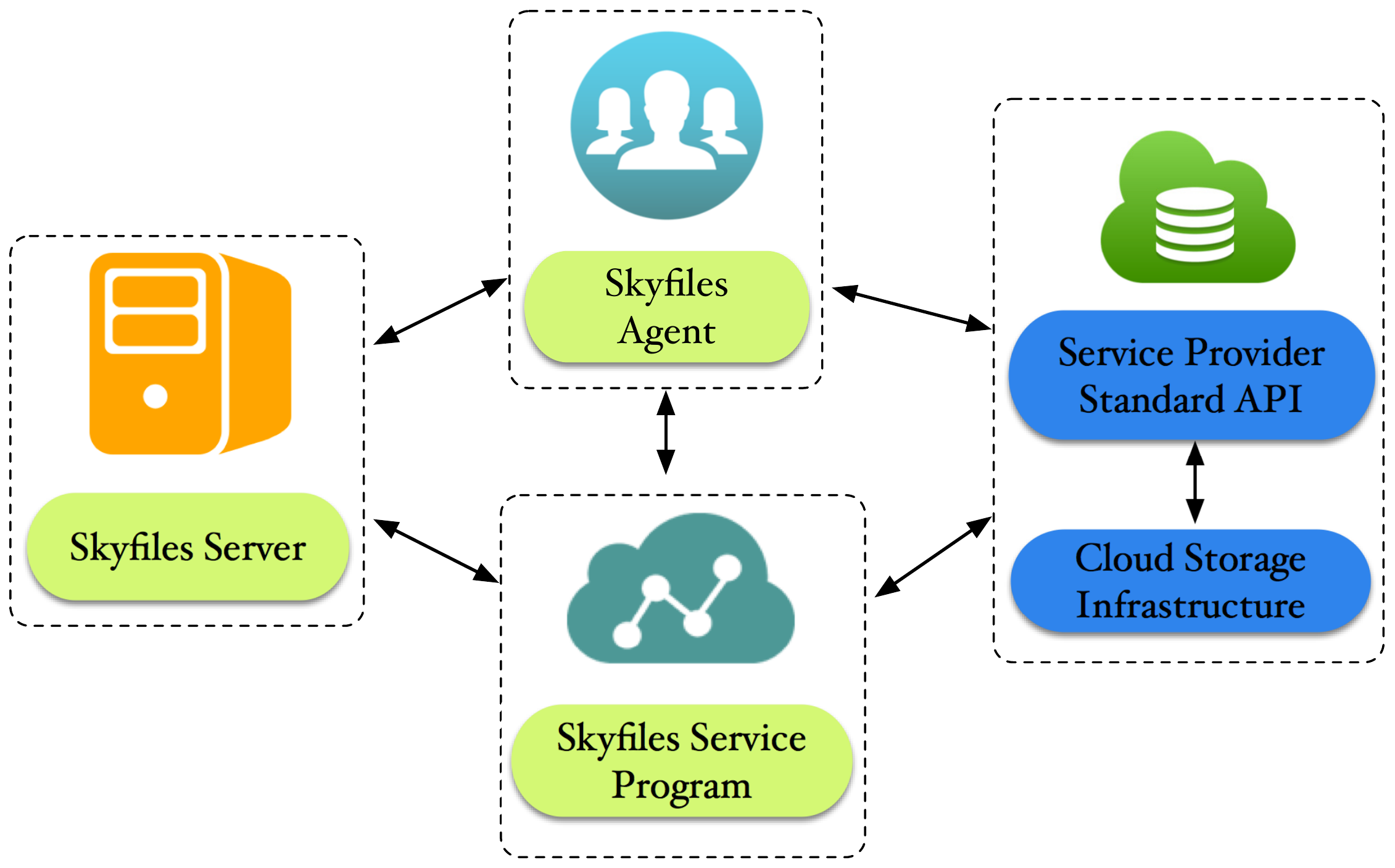}
\caption{Skyfiles system architecture}
\label{fig:skyfile-arch}
\end{figure}


%

\section{File Operations in Skyfiles}
\label{sec:sol2}
In this section, we introduce the detailed implementations of the file operations supported in Skyfiles. We focus on the new set of cloud-assisted file operations and separately present the solution with a private cloud instance and shared cloud instance.

\subsection{With a Private Cloud Instance}
Since a private cloud instance is initiated by the smartphone user, it is considered to be a trusted host for the fie operations.

\subsubsection{Single User File Operations}

Skyfiles processes a single user file operation with the following steps. First, the \texttt{SA} on the phone receives the file operation request, and identifies whether it is a basic file operation or cloud-assisted operation. In the latter case, \texttt{SA} will check if there is an active cloud instance serving the user. If not, \texttt{SA} will launch a private cloud instance using the user's cloud service account. In either case, when a private cloud instance is associated to serve the user, \texttt{SA} will further check if the \texttt{SSP} is running on the cloud instance. If the service is not available, \texttt{SA} will instruct the private cloud instance to contact the Skyfiles Server and fetch the \texttt{SSP}. Once \texttt{SSP} is running on the private cloud instance, \texttt{SA} will establish a secure channel between the smartphone and the cloud instance for the rest of the communication.

At this point, the user request will be analyzed and converted to a sequence of file operation instructions (FOIs). \texttt{SA} will forward the FOIs to \texttt{SSP} on the cloud instance for execution. Once the operation is completed, \texttt{SSP} will send an acknowledgment to \texttt{SA} which will further notify the smartphone user.

\subsubsection{File Transfer between Users}
\begin{figure}[ht]
        \centering
        \subfloat[Conventional Solution]{\includegraphics[width=0.24\textwidth]{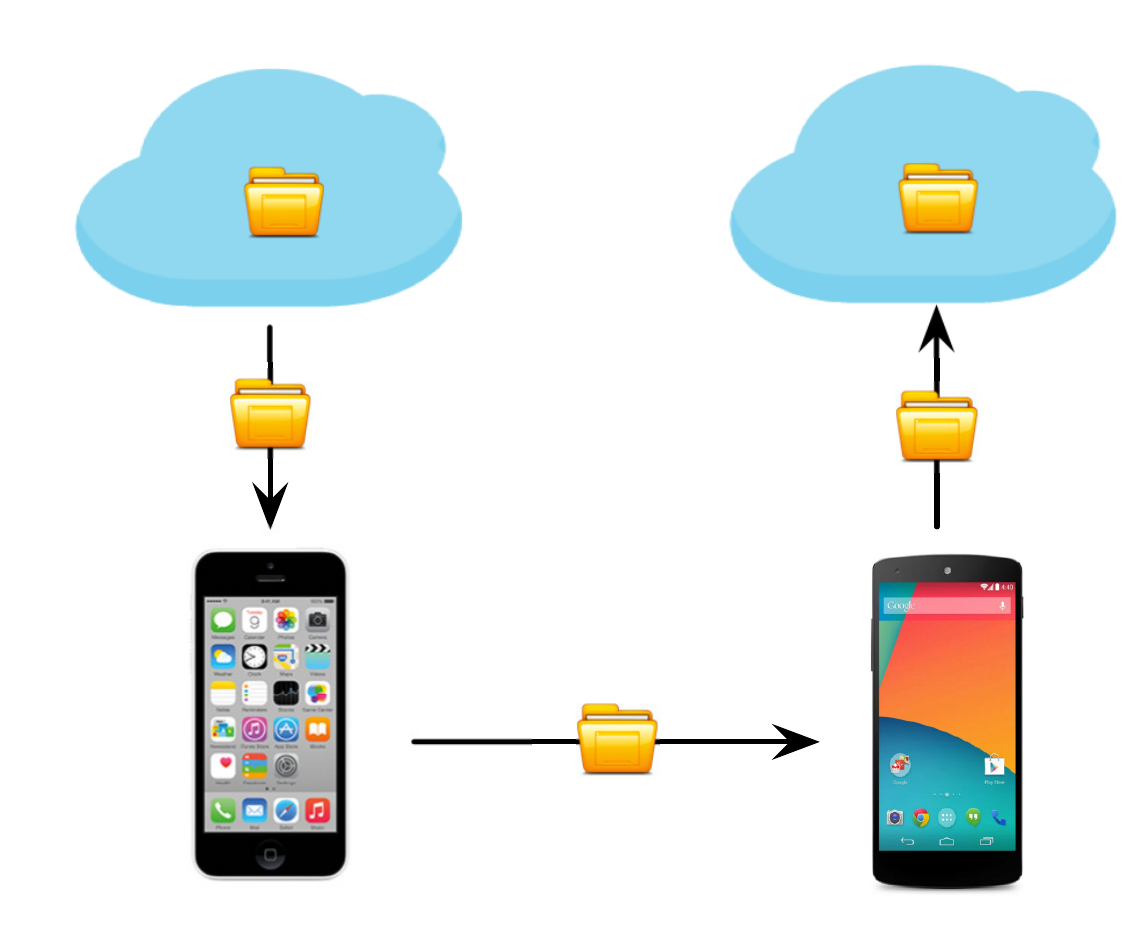}}
        \hfill
        \subfloat[Skyfiles]{\includegraphics[width=0.24\textwidth]{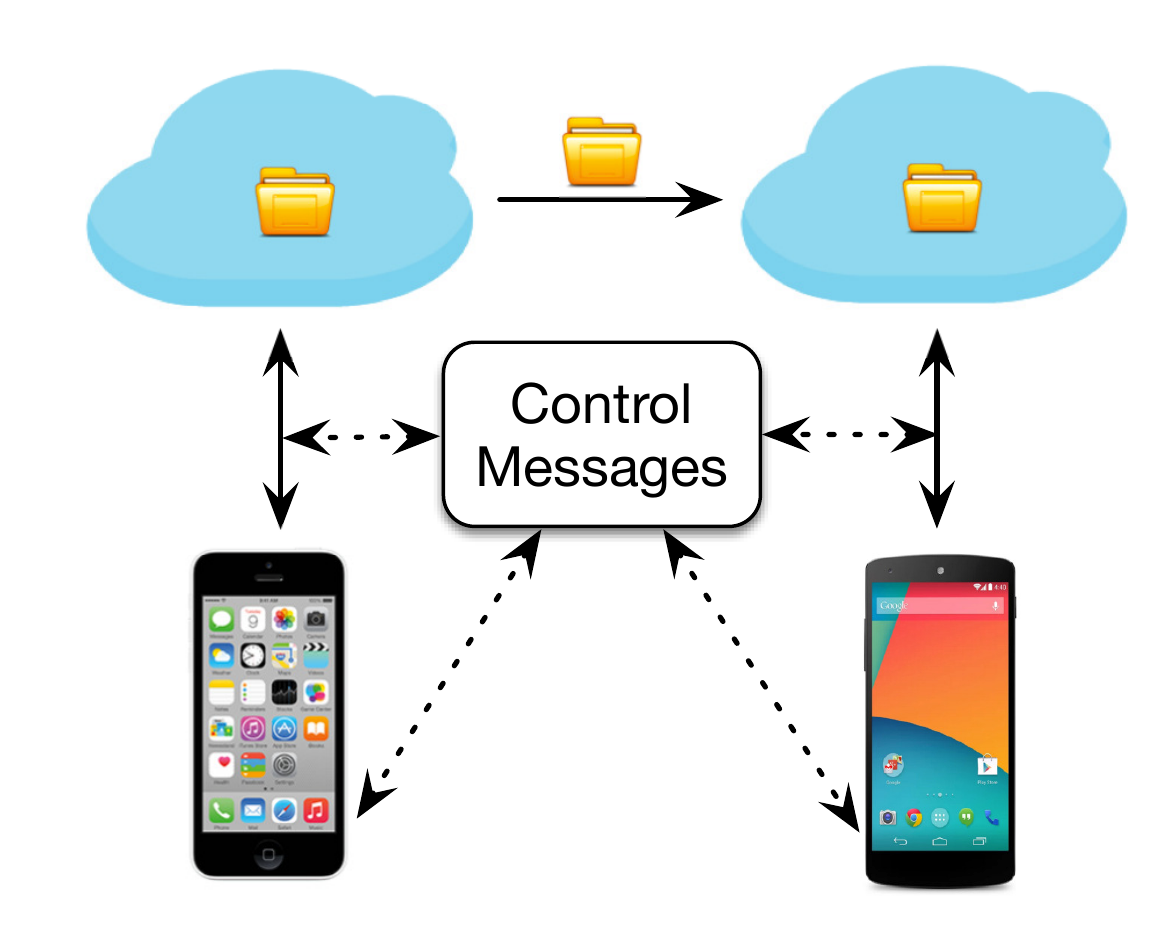}}
        \caption{File transfer between two users}
        \label{fig:demo}
\end{figure}
File transfer between users' cloud storage spaces is also considered as a cloud-assisted operation in Skyfiles.
Basically, a cloud instance is initialized and plays a role of relay node.
It fetches the target files from the sender's cloud storage space and then
forward them to the receiver's cloud storage. In this means, both sender and
receiver's smartphones do not have to hold a local copy of the target files. The bandwidth is consumed only by control messages between smartphones
and the cloud instance/cloud storage server. The basic design of using a
single cloud instance, however, is challenging in practice when no trust has
been established between the sender and receiver. The cloud instance can be
created by either the sender or receiver. In either case, it is not a secure
solution for the party who does not own the instance because the cloud
instance will need to obtain the security credentials of cloud storage from
both sender and receiver to complete the file transfer. Therefore, the owner
of the cloud instance will be able to access the cloud storage space of the
other user which could breach data privacy and lead to other malicious
operations.

To address the above issue, in Skyfiles, we develop a solution that requires
two cloud instances, one from the sender (user \texttt{U\textsubscript{A}}) and the other from the receiver (user \texttt{U\textsubscript{B}}). In the remaining of this paper, we refer to either party in the interactions as user. However, it should note the communication is done using the corresponding \texttt{SA}. Figure \ref{fig:flow-UA-UB} shows the flow of how to
accomplish file transfer between two users (Step 1 to 7). Assume user \texttt{U\textsubscript{A}} is trying to send a file $F$
(or a set of files) to user \texttt{U\textsubscript{B}}. Let $P_{src}$ be the
location of $F$ in the cloud storage of \texttt{U\textsubscript{A}} and $P_{dst}$ be the destination
location that \texttt{U\textsubscript{B}} will put in his cloud storage. In our protocol
description, \texttt{U\textsubscript{A}} and \texttt{U\textsubscript{B}} also represent the users' smartphones. 
First, \texttt{U\textsubscript{A}} starts a cloud instance \texttt{I\textsubscript{A}} and uploads the
security credentials for accessing his cloud storage space to the instance.
\texttt{U\textsubscript{A}}'s request also includes the source file location $P_{src}$ and an {\em
intermediate file location} $URI_F$ (unique resource identifier) which
indicates where $F$ is stored on the instance. The cloud instance will use
the security credentials to download the target file $F$ to its local disk.
At this point, \texttt{I\textsubscript{A}} needs to make $F$ accessible to user \texttt{U\textsubscript{B}}. It first
sends \texttt{U\textsubscript{A}} the intermediate file location $URI_F$. Then, \texttt{I\textsubscript{A}} can set $F$
publicly available or creates a guest account and set the permissions of $F$
so that only the guest account can access it. In the latter case, the
security information for login as the guest account, such as login
password or identity file, needs to be sent back to \texttt{U\textsubscript{A}} as well. Afterwards, the steps on \texttt{U\textsubscript{A}}'s side have been completed. Then, \texttt{U\textsubscript{A}} needs to
notify \texttt{U\textsubscript{B}} necessary information for accessing $F$. Since this step of
communication includes sensitive information, Skyfiles adopts NFC protocol
to securely deliver $URI_F$ and optional login information from \texttt{U\textsubscript{A}}'s
smartphone to \texttt{U\textsubscript{B}}'s phone. At receiver's side, \texttt{U\textsubscript{B}} also starts a cloud
instance \texttt{I\textsubscript{B}} which obtains $F$ from \texttt{I\textsubscript{A}} based on $URI_F$. Finally, \texttt{I\textsubscript{B}}
uploads $F$ to $P_{dst}$. Here, both sender and receiver start a cloud
instance to behave as their agents. The data transfer of $F$ is between
cloud instances and cloud storage servers, which does not consume bandwidth
of users' smartphones. Meanwhile, the security credentials for accessing
cloud storage are only sent to the instance created by the same owner. Thus,
in Skyfiles, file transfer between two users' cloud storage is efficient and
secure.

\begin{figure}[ht]
        \centering
        \includegraphics[width=.49\textwidth]{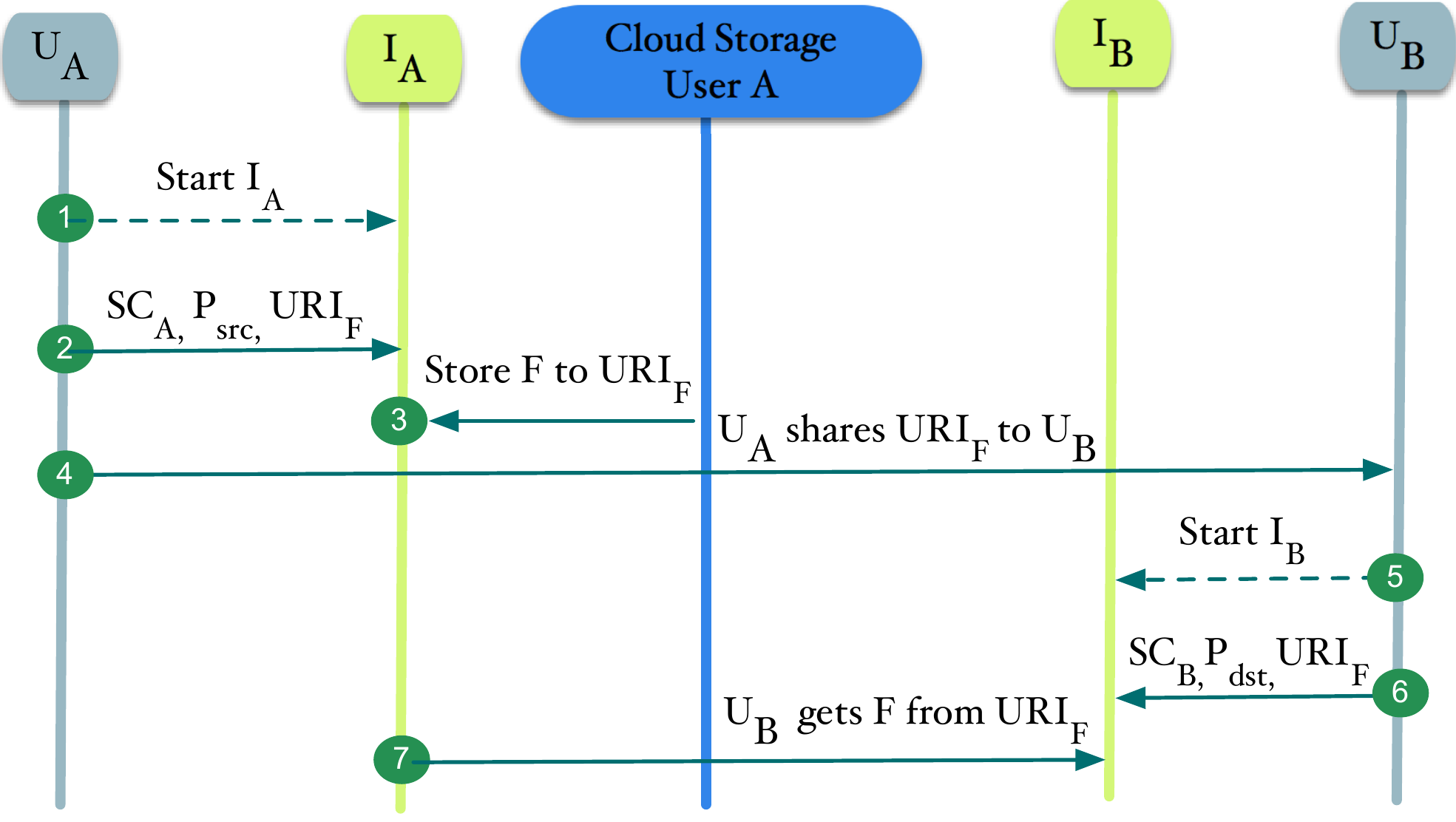}
        \caption{File transfer from \texttt{U\textsubscript{A}} to \texttt{U\textsubscript{B}}}
        \label{fig:flow-UA-UB}
\end{figure}












\subsection{With A Shared Cloud Instance}

In Skyfiles, the other alternative of launching a private cloud instance is to use a shared cloud instance which could save the initial overhead and cost.
While benefiting the performance, the design of sharing instances among users incurs challenge for security. First, it is risky for a
user to upload his security credentials of cloud storage
account to other users' cloud instances. The owner of the
instance may monitor and catch the security credential, and
gain the access to the user's cloud storage space. Second, when
open to public, the shared cloud instances may be used by
malicious users to launch attacks.

In Skyfiles, we have developed a framework for sharing
instances with the involvement of the trusted Skyfiles Server. This server
maintains a list of available cloud instances for sharing and
coordinates the users who request instances and share
instances. Skyfiles applies the following two basic policies to
address the security concerns. First, an instance is shared in
the form of launching \texttt{SSP} as a background service and accepting
requests from other users, rather than allowing other users to
login and execute arbitrary programs. Second, when using a
shared instance, a user does not upload the security
credentials of his cloud storage account in plaintext, but in
an encrypted format. In this way, the owner of the instance can
not gain the access to the tenant user's cloud storage space
and any user's privileges on shared instances are limited to
the specified file operations.

Specifically, there are three types of entities in our design,
the trusted Skyfiles Server \texttt{Serv}, a user \texttt{U\textsubscript{A}} who wants to conduct file
operations on a shared instance, and an available cloud
instance \texttt{I} owned by another user \texttt{U\textsubscript{B}}. The trusted server
holds a binary program \texttt{SSP} and can be running on a shared
instance to provide Skyfiles services to other users. Once a
user (\texttt{U\textsubscript{B}}) decides to share his instance (\texttt{I}), the instance
will contact the server and forwards the basic information
about \texttt{I} such as operating system, hardware setting, and the
time left for sharing. The response from the server is an
executable binary \texttt{SSP}, which has a unique identifier $PID$. After \texttt{SSP} is deployed properly the server will produce an initial key $k_{serv}=H(PID||t)$, where $t$ is the current timestamp rounded to minute and $H(x)$ is an one-way function, such as SHA256. The server then will send $k_{serv}$ and a random time offset $t_{offset}$ (ranging between 1 to 512 seconds) to the program \texttt{SSP} over a secure channel, such as TLS. \texttt{SSP} will send acknowledgment back to \texttt{Serv} and then the instance \texttt{I} is ready to serve requests. Both \texttt{Serv} and \texttt{SSP} will independently produce a new key $k'_{serv}=H(k_{serv}||t+t_{offset})$ after some time intervals (default to 3~min in our design). While $t$ is the current timestamp, $t_{offset}$ is used as a mask to prevent side channel attacks. During an interval, both \texttt{Serv} and \texttt{SSP} preserve $k_{serv}$ and $k'_{serv}$ to handle rare cases in which user data is encrypted under the old key while a new key is generated. This key generation processes ensure both \texttt{Serv} and \texttt{SSP} are able to independently produce a temporal-dependent key for the entire execution-time of \texttt{SSP}. The temporal-dependent key can be used for authentication when \texttt{SSP} want to communicate with \texttt{Serv}, because they can independently generate the same shared secret \cite{EIDgoogle}.

We assume the program \texttt{SSP} is protected by program obfuscation techniques and is validated by typical certificate techniques. After the key exchange processes, \texttt{I} will execute \texttt{SSP} as a service
and be ready to accept other users' requests. The server, on
the other hand, adds \texttt{I\textsubscript{B}} into the list of available instances
for sharing. Finally, each shared instance \texttt{I} can set a
scheduled task to automatically shut down the instance before
the additional charge is incurred. During the shutdown process,
\texttt{I} also notifies \texttt{Serv} which will consequently remove
\texttt{I} from the list of available instances for sharing.

\begin{figure}[ht]
        \centering
        \includegraphics[width=.49\textwidth]{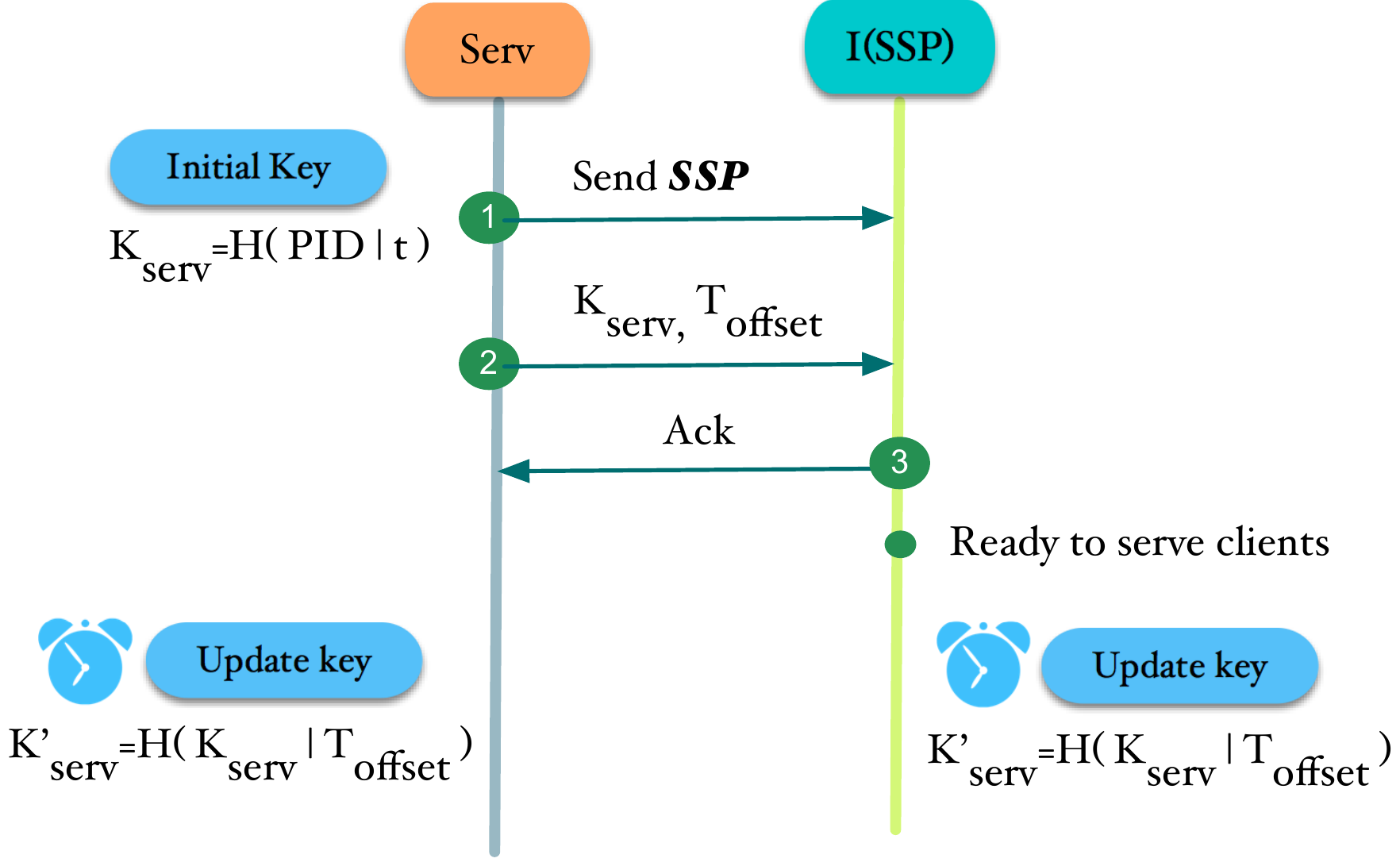}
        \caption{SSP program deployment and key update.}
        \label{fig:flow-SSP-Program}
\end{figure}




\subsubsection{Single User File Operations} 

When a user (\texttt{U\textsubscript{A}}) requests to
use a shared instance to conduct operations on his files on
cloud, he needs to first contact the trusted server \texttt{Serv}. The \texttt{Serv}
will generate a random element $R_A$ for the requesting user and produce a user specific key $k_A$ by a hash function using the current server key $k_{serv}$ and $R_A$, such that $k_A=H(k_{serv}||R_A)$.

The server then chooses a shared instance from the list to
serve \texttt{U\textsubscript{A}} and it could be an interactive process that
involves \texttt{U\textsubscript{A}}'s opinion. Assume \texttt{I} is selected, the server
sends $(R_A, k_A)$ and the IP address of \texttt{I} back to the user
\texttt{U\textsubscript{A}}. Next, \texttt{U\textsubscript{A}} will encrypt the security credentials of his
cloud storage using key $k_A$, $\zeta_A=enc(k_A, SC_A)$, and upload the ciphertext $\zeta_A$ and
$R_A$ to the shared instance \texttt{I}. Since the \texttt{SSP} can independently produce a copy of the $k_{serv}$ following the procedures described earlier, the \texttt{SSP} can generate the user key $k_A$ and decrypt the security credentials of \texttt{U\textsubscript{A}}; such that $SC_A=dec(k_A, \zeta_A)$. In this way, the security credentials are securely transferred to the cloud instance.

\begin{figure}[ht]
        \centering
        \includegraphics[width=.49\textwidth]{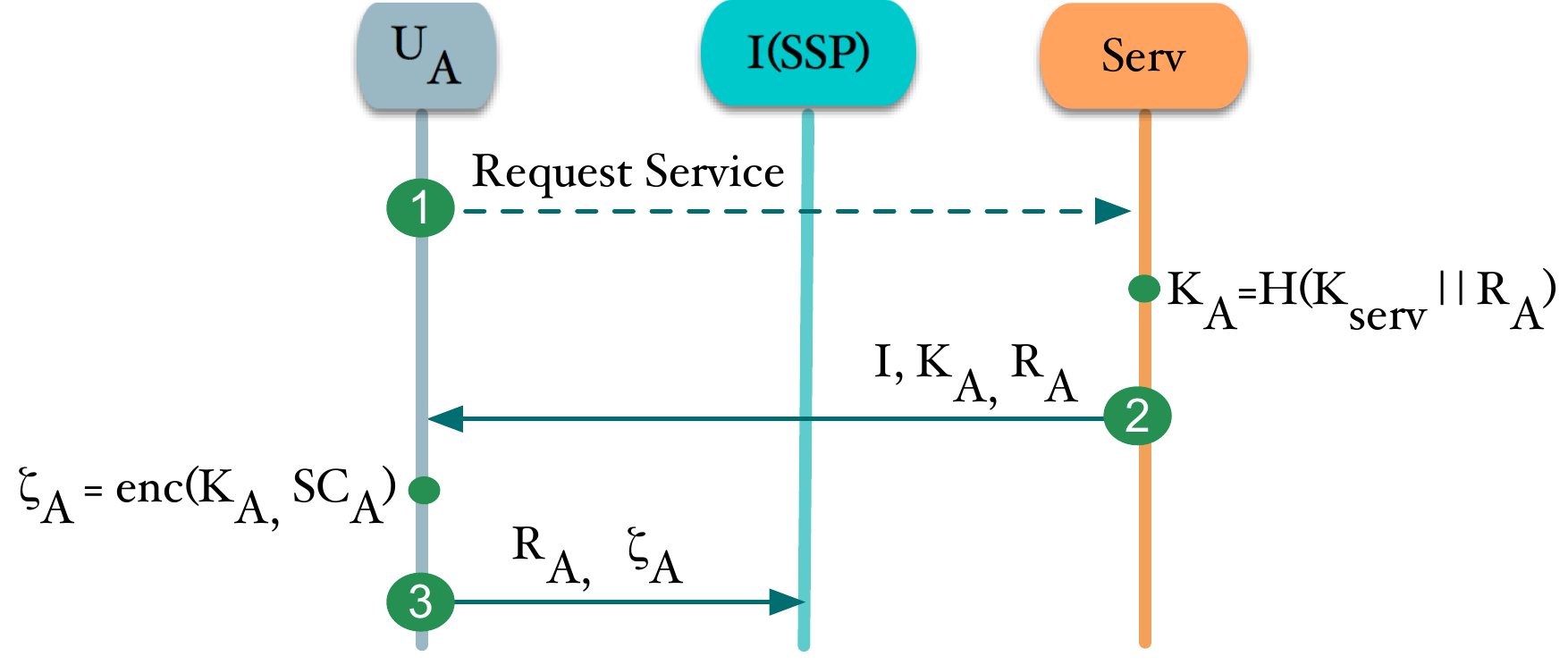}
        \caption{Serve request from \texttt{U\textsubscript{A}}.}
        \label{fig:flow-serve-UA-request}
\end{figure}



It is possible that multiple users share the same instance in which case the server \texttt{Serv} will generate user specific random number $R_i$ and key $k_i$.

\subsubsection{File Transfer Between Users} 
In Skyfiles, two users can
also request a shared instance for transferring files between
their cloud storage. Following the design in
Section~\ref{sec:sol}, the sender will initialize the process
and request a shared instance from the server. Compared to the
single user operations, file transfer requires both sender
(\texttt{U\textsubscript{A}}) and receiver (\texttt{U\textsubscript{B}}) to send the security credentials
of their cloud storage account to the shared instance. In
addition, the sender needs to notify the receiver the instance
assigned by \texttt{Serv}. Figure \ref{fig:flow-file-transfer-UA-UB} shows the major
messages exchanged in our design. 

\begin{figure}[ht]
        \centering
        \includegraphics[width=.49\textwidth]{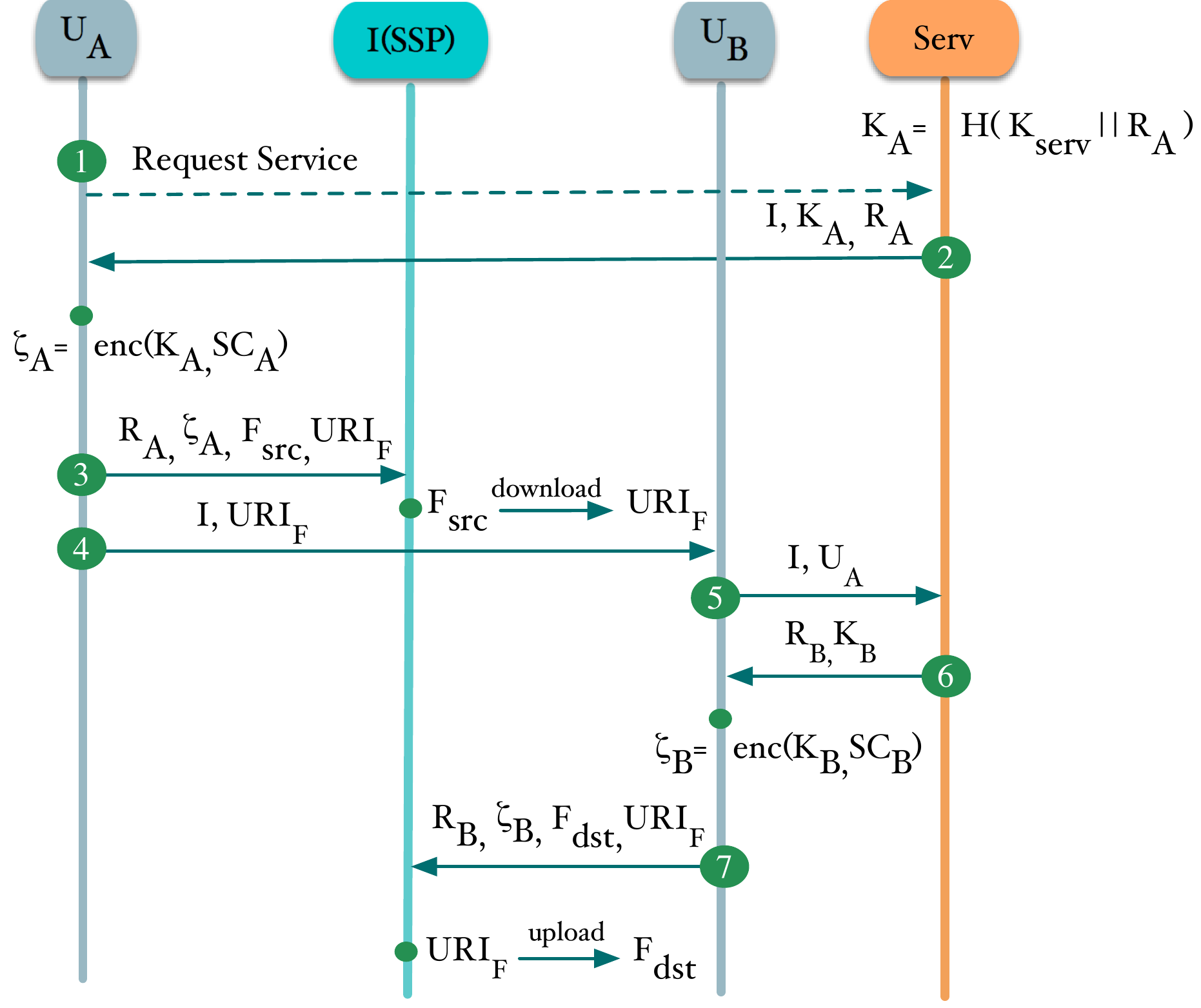}
        \caption{File transfer between \texttt{U\textsubscript{A}} and \texttt{U\textsubscript{B}}.}
        \label{fig:flow-file-transfer-UA-UB}
\end{figure}

%
%
%
%
%
%
%
%
%
%
%
%














Sender \texttt{U\textsubscript{A}} initializes the request by contacting the server
\texttt{Serv}. 
\texttt{U\textsubscript{A}} will upload the encrypted security credentials $\zeta_A$ to \texttt{I} as well as the
source file location ($F_{src}$) and intermediate file location ($URI_F$). Then \texttt{U\textsubscript{A}} will notify the
receiver \texttt{U\textsubscript{B}} the shared instance \texttt{I} and the location of the
target files ($URI_F$). This message is attached with the
certificate from \texttt{Serv} so that the receiver can verify the shared
instance \texttt{I} is legitimate. Next, the receiver \texttt{U\textsubscript{B}} sends
the server a request with (\texttt{U\textsubscript{A}}, \texttt{I}). After verifying
there exists a shared instance \texttt{I} serving \texttt{U\textsubscript{A}}, the server
will send back $R_B$ and $k_B$ to \texttt{U\textsubscript{B}} so that \texttt{U\textsubscript{B}} can encrypt his security credentials in the same way as \texttt{U\textsubscript{A}}. Eventually, \texttt{U\textsubscript{B}}
uploads the encrypted security credentials $\zeta_B$ and the intermediate file location ($URI_F$) and destination file location ($F_{dst}$) to the shared instance \texttt{I}.


\subsection{Shared Cloud Instance Management}
In the above subsection, we present the design to securely utilize the shared cloud instance by using the Skyfiles Service Program. In our design, the Skyfiles Server manages a pool of shared instances. Since the cloud instances are provided by various users, the features of the instances could be different in terms of starting time, ending time and bandwidth, etc. In addition, the types users' requests could be different as well. For example, \texttt{U\textsubscript{A}} wants to "zip" a picture while \texttt{U\textsubscript{B}} and \texttt{U\textsubscript{C}} would like to transfer a large file between their cloud storages. In this case, \texttt{U\textsubscript{A}}'s request need less time and bandwidth than the file transfer request. In the subsection, we discuss on the management problem of the shared instance pool. 

Formally, the 
problem is defined as follows:
Given a set of cloud instances \texttt{C} and for each cloud instance $c(\tau_s, \tau_e, \beta) \in$ \texttt{C}, $\tau_s$ and $\tau_e$ indicates its service starting time and ending time, $\beta$ stands for its total free bandwidth. Suppose there are a set of tasks $K$, and for each request $k(\alpha, \gamma, \delta ) \in K $, $\alpha$ and $\gamma$ symbolize its requested starting time and ending time, and $\delta$ means its requested bandwidth. The problem is to allocate all the tasks by using minimum number of cloud instances without exceeding cloud instance's bandwidth limit, such that for each accommodated request, its requested starting and ending time is within the feasible working time of its allocated cloud instance and its requested bandwidth is satisfied. 
In this problem, we assume that one task cannot be split into sub-tasks and executed on different cloud instances. 

\begin{theorem}
The problem is NP-hard.
\end{theorem}

\begin{proof}
Let's first introduce the Bin-Packing problem: Given $n$ items with sizes $e_1$, $e_2$, \ldots , $e_n$, and a set of $m$ bins with capacity $c_1$, $c_2$, \ldots ,$c_m$, the Bin-Packing problem is to pack all the items into minimized number of bins without violating the bin capacity size. 
Assume that all the tasks have the same requested starting time $\alpha$ and ending time $\gamma$ , and all the cloud instances also have the same starting time $\tau_s$ and ending service time $\tau_e$, such that $\alpha \geq \tau_s$ and $\gamma \leq \tau_e$. Consequently, this problem is equivalent to the bin-packing problem, which is NP-hard \cite{Garey79}.   
\end{proof}

Next, we present an Integer Linear Program (ILP) to solve this problem.

\textbf{Notations:}
\begin{itemize}
\item $k(\alpha, \gamma, \delta)$: a task specifies the starting time $\alpha$, ending time $\gamma$ and requested bandwidth $\delta$. $\gamma-\alpha$ can be interpreted as the total task executing time.
\item $K$: a set $|k|$ tasks.
\item $c(\tau_s, \tau_e, \beta)$: a cloud instance that provides service from time $\tau_s$ to $\tau_e$ with total free bandwidth of $\beta$.

\item \texttt{C}: A set of $|$\texttt{C}$|$ shared instances.  

\end{itemize}

\textbf{Variable:}
\begin{itemize}
\item $Y_k^c$: boolean variable indicating whether task $c$ has been assigned by instance $c$. 
\end{itemize}

\textbf{Objective:}
\begin{equation} \label{Eq:Obj}
\min  \sum_{c \in \texttt{C}}  \max_{k \in K} Y_k^c
\end{equation}

\textbf{Bandwidth constraint:}

\begin{flalign} \label{Eq:AC1}
 \sum_{k(\alpha, \gamma, \delta) \in K:  \alpha \geq \tau_s \&\& \gamma \leq \tau_e} \delta \cdot Y^c_k  \leq \beta  ~~\forall c(\tau_s,\tau_e,\beta) \in \texttt{C}
\end{flalign}

\textbf{Assignment constraint:}

\begin{flalign} \label{Eq:AC2}
 \sum_{c \in \texttt{C}} Y_k^c = 1  ~~\forall k(\alpha, \gamma, \delta) \in K
\end{flalign}

Eq.~(\ref{Eq:Obj}) is the objective function which tries to minimize the number of used cloud instances. Eq.~(\ref{Eq:AC1}) ensures that all the accommodated tasks on a certain server should not exceed its freed bandwidth such that for each task its starting time and ending time should be within the cloud instance's server time. Eq.~(\ref{Eq:AC2}) ensures that each requested task must be assigned with one cloud instance.


\section{Performance Evaluation}
\label{sec:eval}

We have implemented Skyfiles system on Android with Dropbox~\cite{dropbox}
storage service and tested it on Google Nexus smartphone. 
For cloud-assisted operations, we use the service provided by Amazon
Web Service (AWS EC2)~\cite{aws}. All the experiments are conducted on
the Micro instance (613 MB memory, up to 2 EC2 Compute Units and up to 15GB bandwidth). 
Fig~\ref{implementation} shows the user interface of Skyfiles.
Basically, Skyfiles fetches meta data of user's Dropbox files to construct the shadow file system 
and when user clicks one file, the available operations highlight out. Users can login with their AWS 
accounts and specify whether they would like to share the cloud instances with others.

The major
performance metrics we consider are time overhead
and bandwidth consumption. 
In order to measure the bandwidth consumption, we connect the smartphone to
a WiFi access point (AP) for the Internet access. For a particular file
operation, we start a tcpdump~\cite{tcpdump} session on the AP to record all
data packets sent from and received by the smartphone. Then, we use a
typical analyzer to retrieve the transferred data size from the trace
records. To make the measurement more accurate, we additionally set the
firewall on the smartphone to block network traffic from all the
applications except Skyfiles. In particular, we use Linksys WRT54GL wireless
router running DD-WRT~\cite{ddwrt} in our tests. Finally, the performance
data reported in the rest of this section is the average value of the five
independent experiments with the same configuration.

\begin{figure}[ht]
\centering
\includegraphics[height=0.45\textwidth]{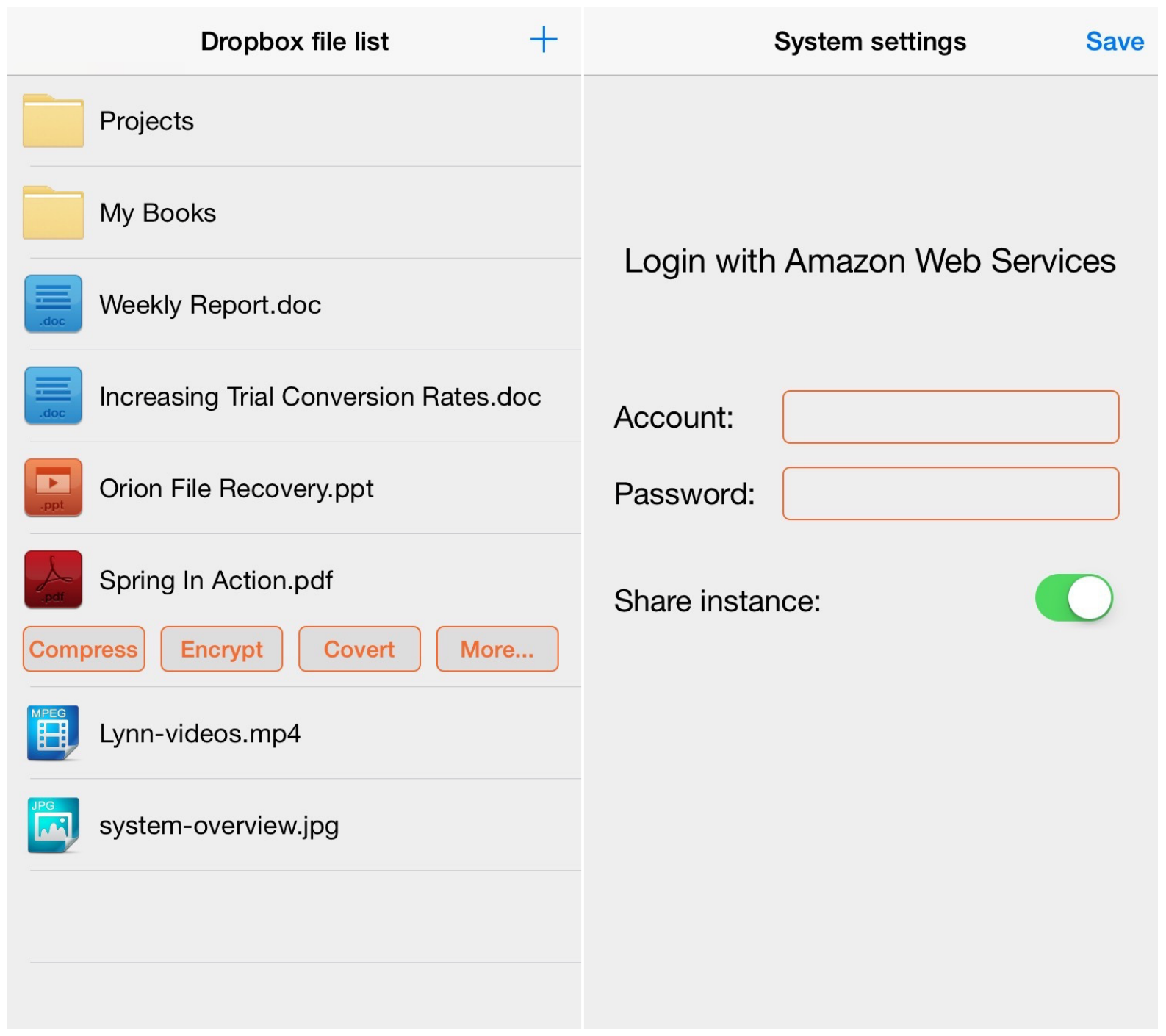}
\caption{Skyfiles screenshots}\label{implementation}
\end{figure}

Note that for the advanced file operations and file transfer between users,
there are rare solutions to compare with because Skyfiles represents
a new system paradigm for cloud storage management. In the rest of this
section, we will solely compare it to the simple solution that requires the
smartphone to download a local copy for those functions.



\subsection{Basic File Operations}
We first present the bandwidth consumption of basic file operations
implemented by Dropbox APIs. In this test, we create a new Dropbox account
with a folder containing 1000 text files (22 bytes each). The operations we
will test are (1) log in Dropbox; create (2) / delete (3) a folder (under the root
directory); create (4) / delete (5) / rename (6) a file; (7) enter/leave a folder.

\begin{figure*}[htbp]
\centering
\begin{minipage}[t]{0.32\linewidth}
\centering
\includegraphics[width=\linewidth]{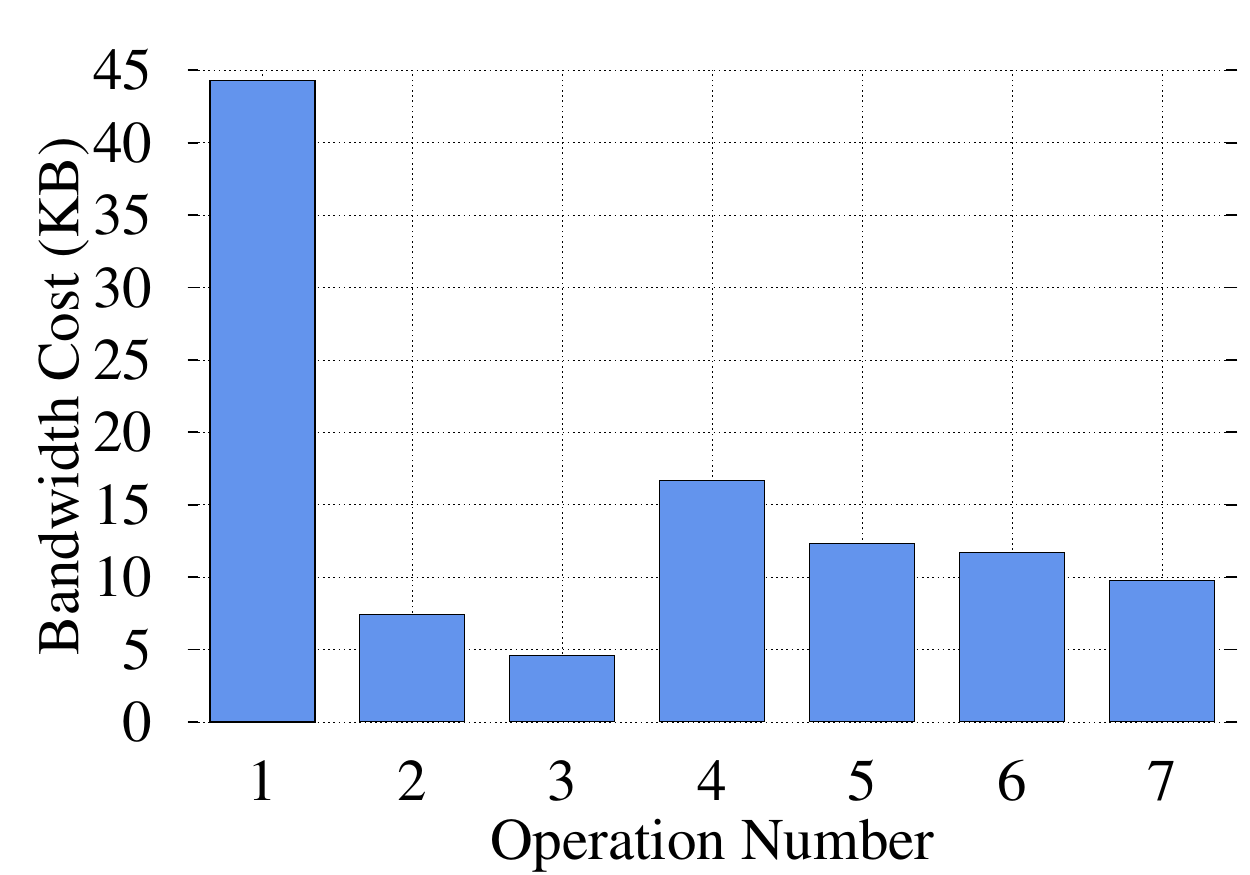}
\caption{Bandwidth cost of basic operations }
\label{basic}
\end{minipage}
~
\begin{minipage}[t]{0.32\linewidth}
    \centering
\includegraphics[width=\linewidth]{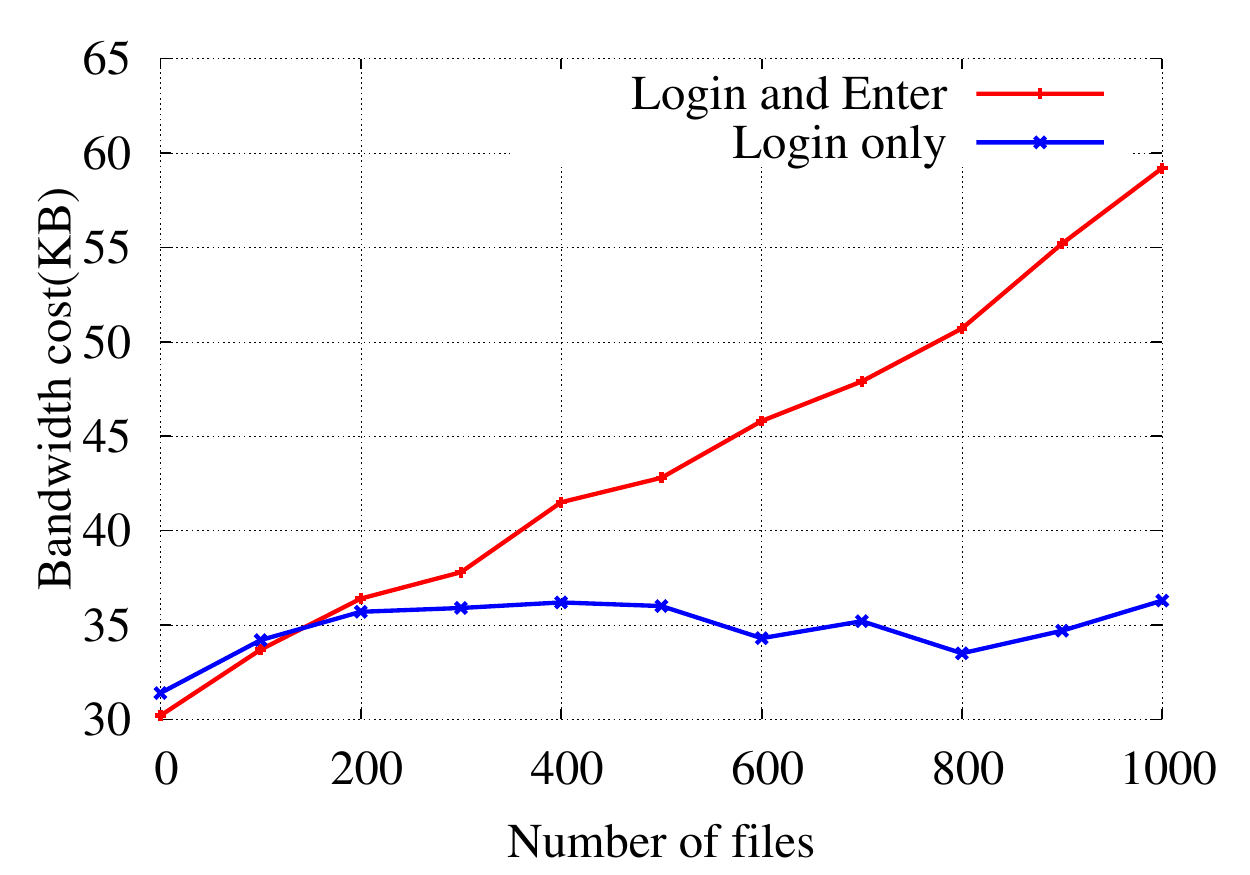}
\caption{Bandwidth cost v.s. number of files}\label{compare}
\end{minipage}
~
\begin{minipage}[t]{0.32\linewidth}
    \centering
\includegraphics[width=\linewidth]{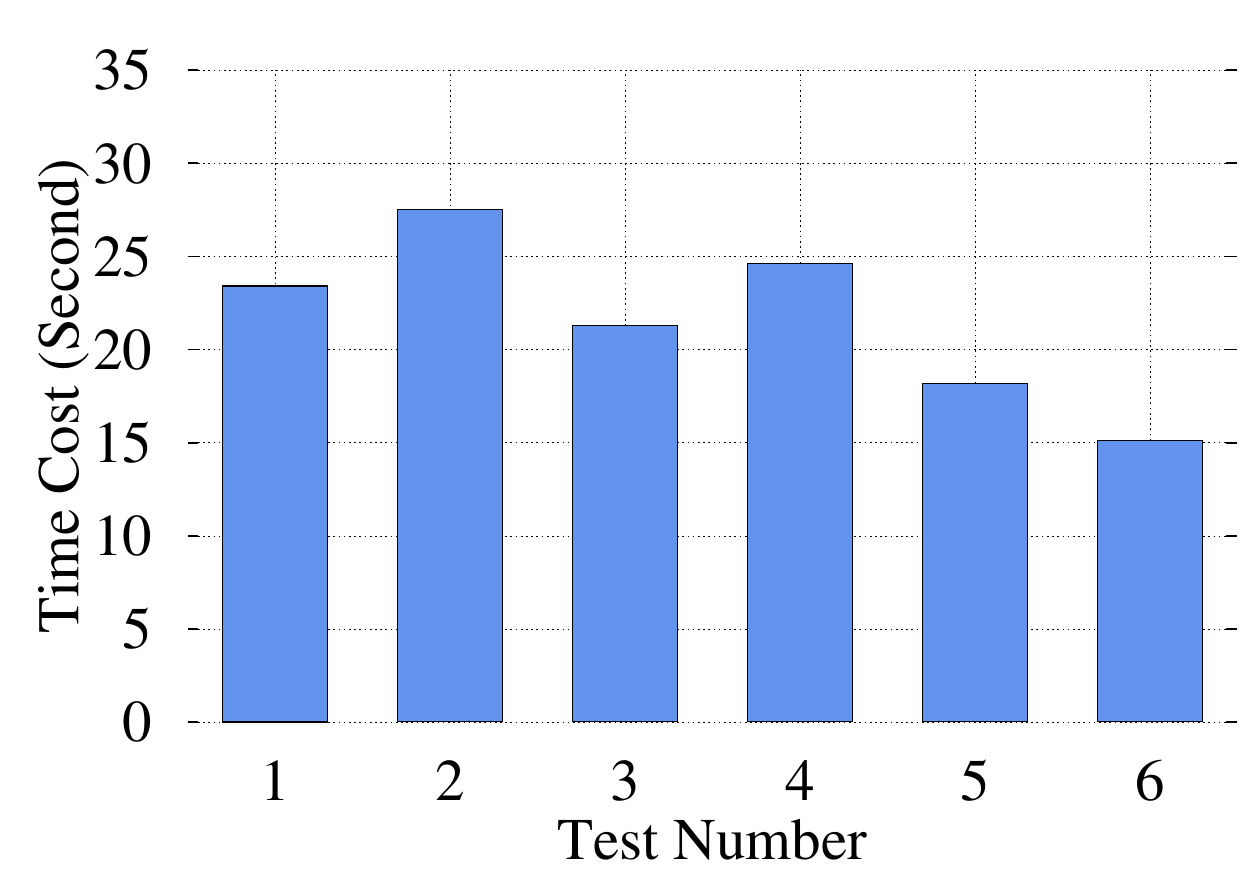}
\caption{Time cost of starting an instance}\label{start-instance}
\end{minipage}
\vspace{-0.1in}
\end{figure*}

For each file operation, Dropbox server requires security credential to be
attached and the communication to be based on SSL. As shown in Fig.\ref{basic},
login process consumes the most bandwidth because the interaction of
authentication and Dropbox APIs needs to recursively fetch meta data to
synchronize/update the local shadow file system. Creating and deleting an
empty folder consumes 7.3KBand 3.9KB respectively, which are the minimum
costs among the tested operations. Creating and deleting a text file is
similar to the previous case. When tested in the folder 'test',  it
certainly incurs more bandwidth cost (15.5KB and 11.2KB). It can be explained by the fact that
the folder contains 1000 other files and that once a change is made in the
folder, Dropbox APIs will re-fetch the list of files in it. Finally, when a
user enters the folder and then leaves, it costs 9.7KB bandwidth, which is
slightly lower than creating/deleting a file. The primary factor is
still the fetch of the entire file list. Fig~\ref{compare} illustrate the detailed bandwidth cost
of login-only and login-with-enter ('test' folder) operations.

\subsection{Cloud-assisted Advanced File Operations}
In this subsection, we evaluate the cloud-assisted file operations in
Skyfiles, particularly downloading and compressing. Due to the page limit, we omit
the results for encrypting and converting operations. The performance of encrypting
operation is similar to that of compressing operation. Skyfiles supports these
operations with user-created cloud instances and shared instances. The
difference among the performances is that using shared instances saves the
initial cost (mainly the time overhead) of starting a cloud instance.
Therefore, we first present the overhead of starting a cloud instance,
then show the performance of these operations assuming that a cloud instance
has been available. The workload we use for testing includes 4 sets of
files. In the rest of ~\ref{sec:eval}, we use file numbers to indicate the tested files: 
\begin{enumerate}
 \item One picture with high resolution option (16MB).
 \item Five pictures taken by Canon Powershot G11 and stored in a .tar file (83MB).
 \item 40 seconds video (MPEG 4) recorded by Samsung Nexus 4 (63MB).
 \item 82 seconds video (MPEG 4) recorded by Samsung Nexus 4 (127MB).
\end{enumerate}

{\bf Overhead of starting a cloud instance:}
We conduct six groups of tests in this experiment at different time of a
day. Each group contains 5 individual operation of starting a AWS Micro
instance. The operation ends when the user is able to log in the instance.
The following Fig.~\ref{start-instance} illustrates the results of the
average value. Overall it is a time-consuming process as
the all tested cases spend more than 14 seconds in starting an instance. 

\begin{table}[ht]
\centering
\begin{tabular}{|l|c|c|c|c|c|c|}
\hline
&Max&Min&&Max&Min\\
\hline
Test 1&19.8&26.7&Test 2&24.4&32.3\\
\hline
Test 3&17.2&325.0&Test 4&20.3&26.5\\
\hline
Test 5&14.3&22.1&Test 6&14.2&15.7\\
\hline
\end{tabular}
\caption{Max/Min overhead of starting an AWS cloud instance (second)}
\label{tab:start}
\end{table}

Amazon web service spends most of the time in allocating necessary
resources, like CPU, memory and storage, for the instance and 
installing software, such as operating system and ssh service. 
Thus, the time cost depends on the amazon web service system workload.
Table~\ref{tab:start} shows the variances of starting an instance in different time slot of a day. Test 6 that we conduct
at 3am-4am achieves smallest variance. According to the AWS usage history report~\cite{aws-price}, the system workload is
usually lower at the 3am-4am.

\begin{figure*}[htbp]
\centering
\begin{minipage}[t]{0.32\linewidth}
\centering
\includegraphics[width=\linewidth]{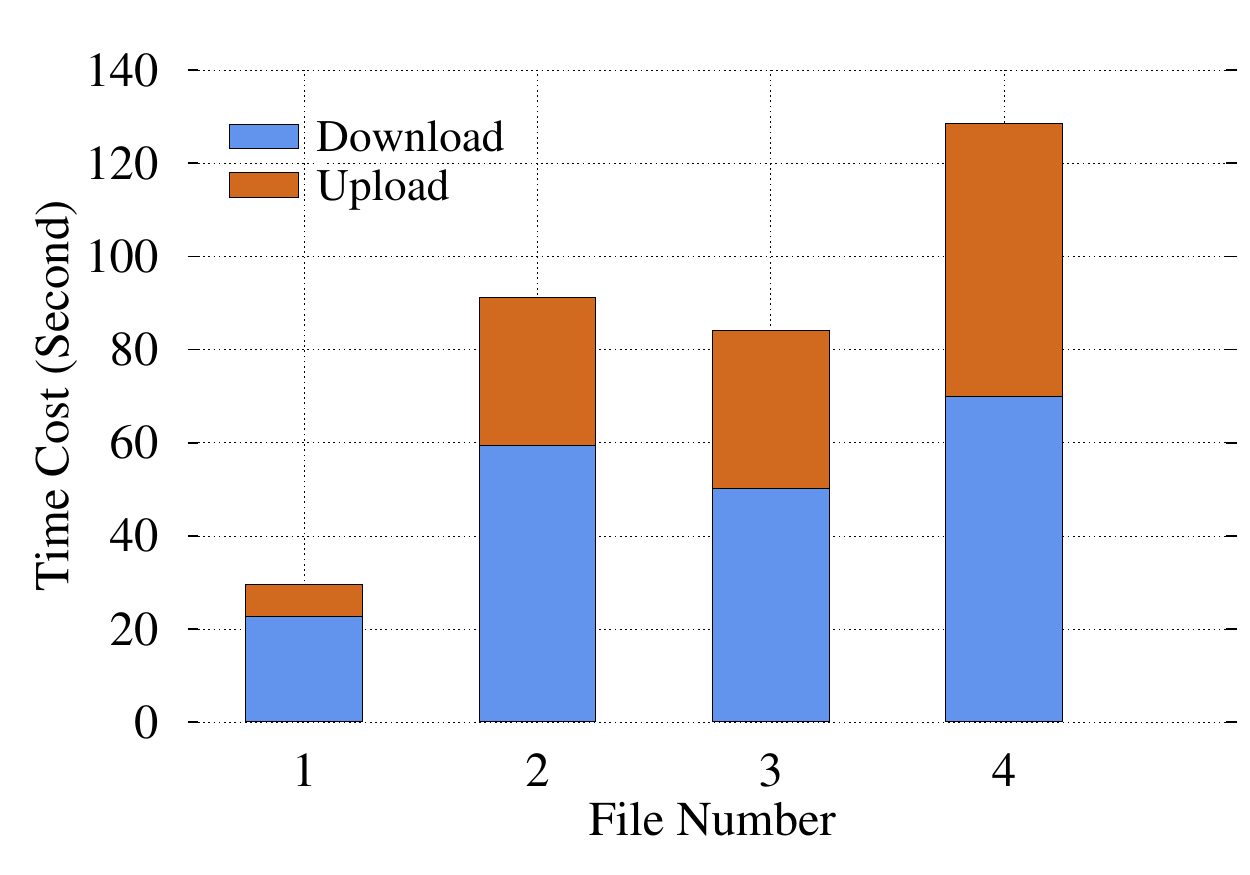}
\caption{Time cost of downloading Files}\label{download_wes}
\end{minipage}
\centering
\begin{minipage}[t]{0.32\linewidth}
    \centering
\includegraphics[width=\linewidth]{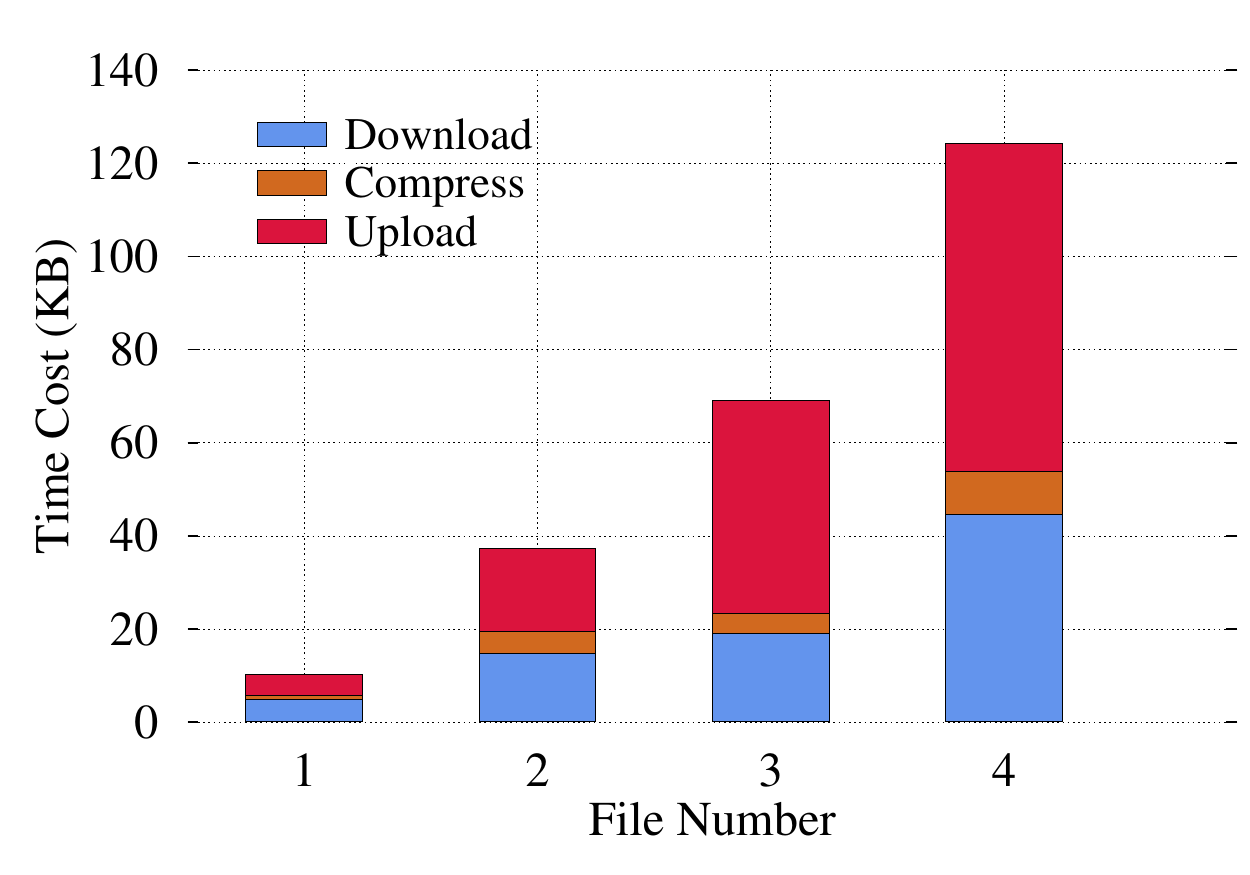}
\caption{Time cost of compressing files}\label{zip_file}
\end{minipage}
\begin{minipage}[t]{0.32\linewidth}
    \centering
\includegraphics[width=\linewidth]{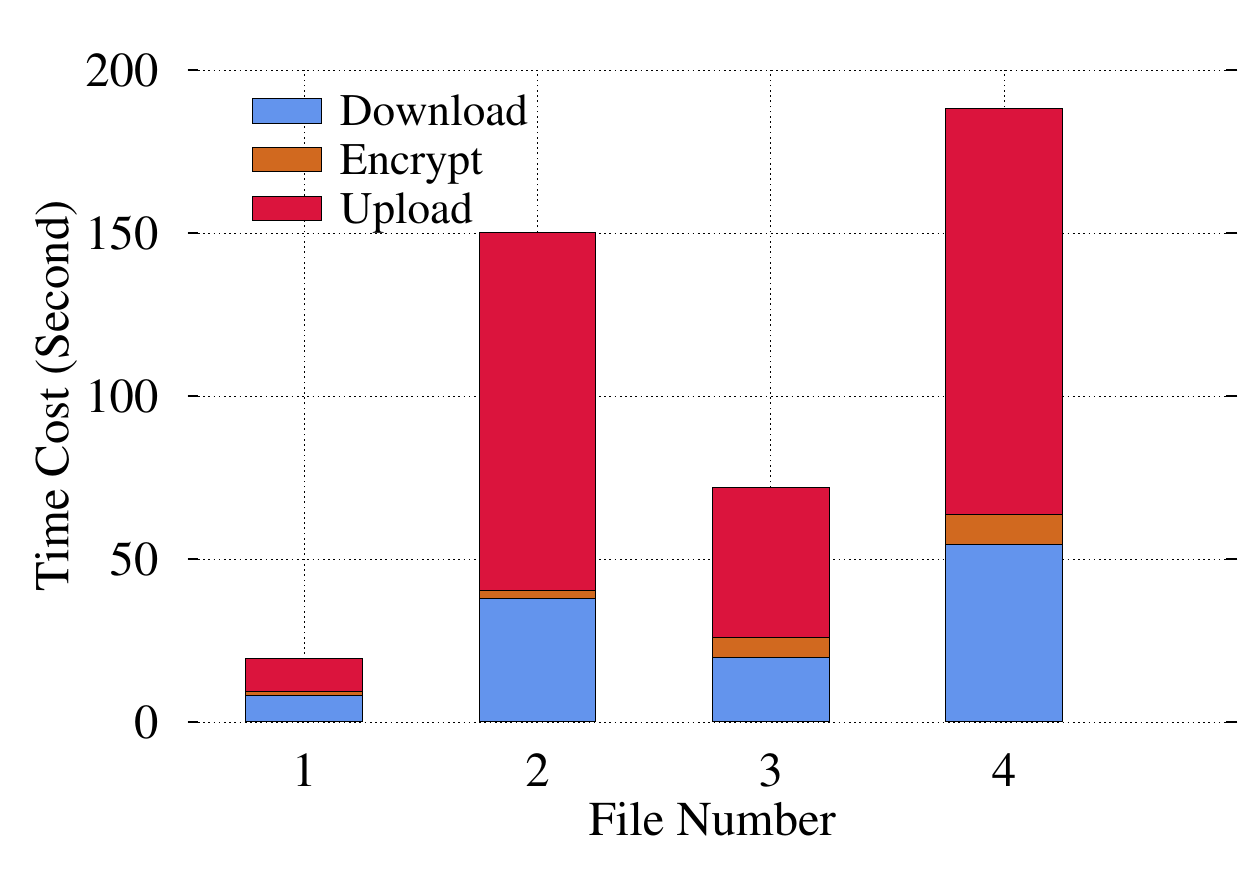}
\caption{Time cost of encrypting files }\label{encrypt_file}
\end{minipage}
\end{figure*}

In the rest of this subsection,
the performance overhead does not include the initial phase of starting an
instance, which is for the case of using shared instance. If the user starts
his own cloud instance, the extra overhead could range from 15 seconds to 30
seconds based on Fig.~\ref{start-instance}.

{\bf Overhead of downloading operation:} In this experiment, we let the cloud
instance download the files in our workload and upload them to our Dropbox
storage space. The target files are hosted in one of our servers. As Fig.~\ref{download_wes} shows, upload/download overhead is roughly
proportional to the file size. Uploading is faster than downloading because
the instance we use (AWS EC2) and the Dropbox service (AWS S3) belong to the
same cloud service provider. Overall, the transmitting rate is around 1MB
per second, which is much faster than downloading files to the
smartphone and then uploading them to Dropbox cloud storage via cellular
network.


{\bf Overhead of compressing operation:} In this experiment, we test
compression operations on Dropbox files. Particularly, we use gzip to
compress the files downloaded to the cloud instance and then upload the
compressed file back to Dropbox cloud storage. Fig.\ref{zip_file} depicts the
breakdown time overhead of this operation. Uploading costs the most, followed
by downloading and compression process. Downloading becomes faster than
uploading because the source files are hosted on Dropbox. This operation in
Skyfiles is faster compare to compressing files locally stored on smartphones.
For example, compressing one picture (16M) and 5 pictures (83M) costs 10.4s
and 38.7s in total. The compressed files in these two cases are 7.7MB and
40.0MB.

{\bf Overhead of encrypting operation:} In this experiment, we use 
OpenSSL~\cite{openssl} at the cloud side to
encrypt the targeted file with Triple DES algorithm. The overhead to complete 
the encryption is shown on Fig.~\ref{encrypt_file}.
As indicated by Fig~\ref{zip_file} and \ref{encrypt_file}, the encryption operation is more time consuming than the compression.
One reason lies in the fact that, after the compression, the file size is reduced. For example, the size of File Number 2 is reduced from 83MB to 40MB.
Another reason resulting in the longer processing time is that the encryption with Triple DES generates two files, one encrypted file
and one key file. After finishing the encryption at the cloud side, Skyfiles uploads the encrypted file to Dropbox and pushes the 
key file which is required to decrypt the file to the mobile device.
%

{\bf Bandwidth consumption:} The bandwidth consumption of advanced file
operations are similar as those of only control messages  exchanged.
Fig.~\ref{download_bandwidth},~\ref{compress_bandwidth},~\ref{encrypt_bandwidth} illustrate
the detailed bandwidth consumption of downloading, compression and encryption. 
Taking compression operation in Fig~\ref{compress_bandwidth} as an example, uploading cost including control messages is very
small, specifically 3.3KB, 4.8KB, 5.1KB and 5.8KB. The uploading bandwidth costs are 
similar on different files. However, Fig~\ref{zip_file} indicates that the uploading time varies among those files, 
4.6s, 17.9s, 45.8s, 70.4s, respectively. The bandwidth is not 
increasing along with the time. It can be attributed to that most of the uploading bandwidth is
consumed at the beginning of the process that is used 
to send command to login and control the cloud server.
The downloading
bandwidth varies across different file sizes. Most of the bytes are consumed by our
periodical heartbeat messages reporting the status of the operation. 
In Fig~\ref{encrypt_bandwidth}, the total bandwidth consumption is always lower than 300KB in our tests,
regardless of the target file size. In the conventional simple solution, the
smartphone has to obtain a local copy to support the advanced operations.
The bandwidth consumption, thus, is roughly double the target file size.

\begin{figure*}
\begin{minipage}[t]{0.32\linewidth}
\centering
\includegraphics[width=\linewidth]{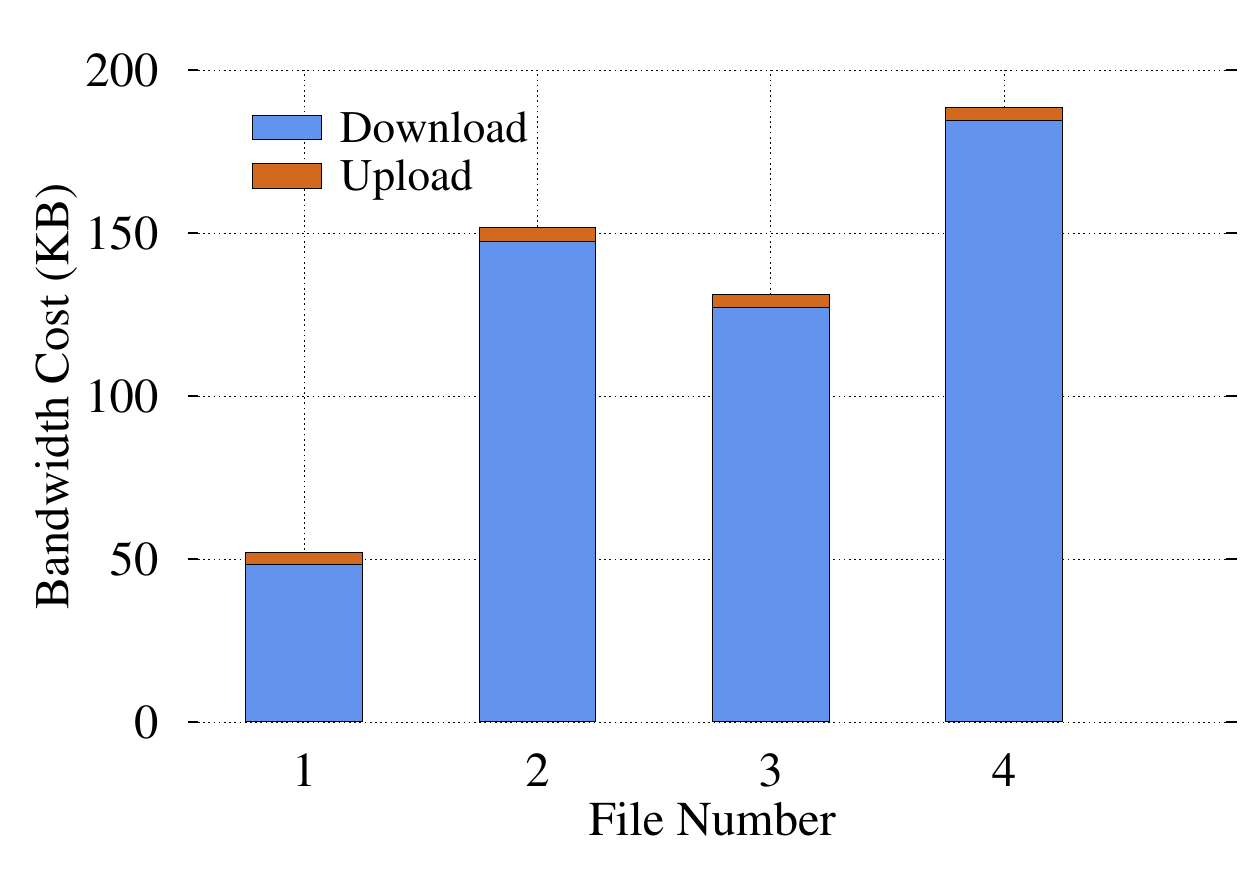}
\caption{Bandwidth cost of downloading files }\label{download_bandwidth}
\end{minipage}
\begin{minipage}[t]{0.32\linewidth}
\centering
\includegraphics[width=\linewidth]{figs/skyfile_compress_time-eps-converted-to.pdf}
\caption{Bandwidth cost of compressing files }\label{compress_bandwidth}
\end{minipage}
\begin{minipage}[t]{0.32\linewidth}
\centering
\includegraphics[width=\linewidth]{figs/skyfile_encrypt_time-eps-converted-to.pdf}
\caption{Bandwidth cost of encrypting files}\label{encrypt_bandwidth}
\end{minipage}
\end{figure*}

\subsection{File transfer between Users}
The performance of file transfer is similar to the above advanced file
operations. The process is identical to the downloading phase in {\em
compress operation} followed the uploading phase in {\em download
operation}, i.e., target files are downloaded by an instance from Dropbox
storage (sender's account) and then uploaded back to Dropbox storage
(receiver's account). The time overhead and bandwidth consumption are very
close to those in advanced file operations such as compression and encryption.
When using shared instances, the extra costs for communicating with the
trusted server is negligible compared to the total performance.

\begin{figure}[htbp]
\centering
\begin{minipage}[t]{0.48\linewidth}
\centering
\includegraphics[width=\linewidth]{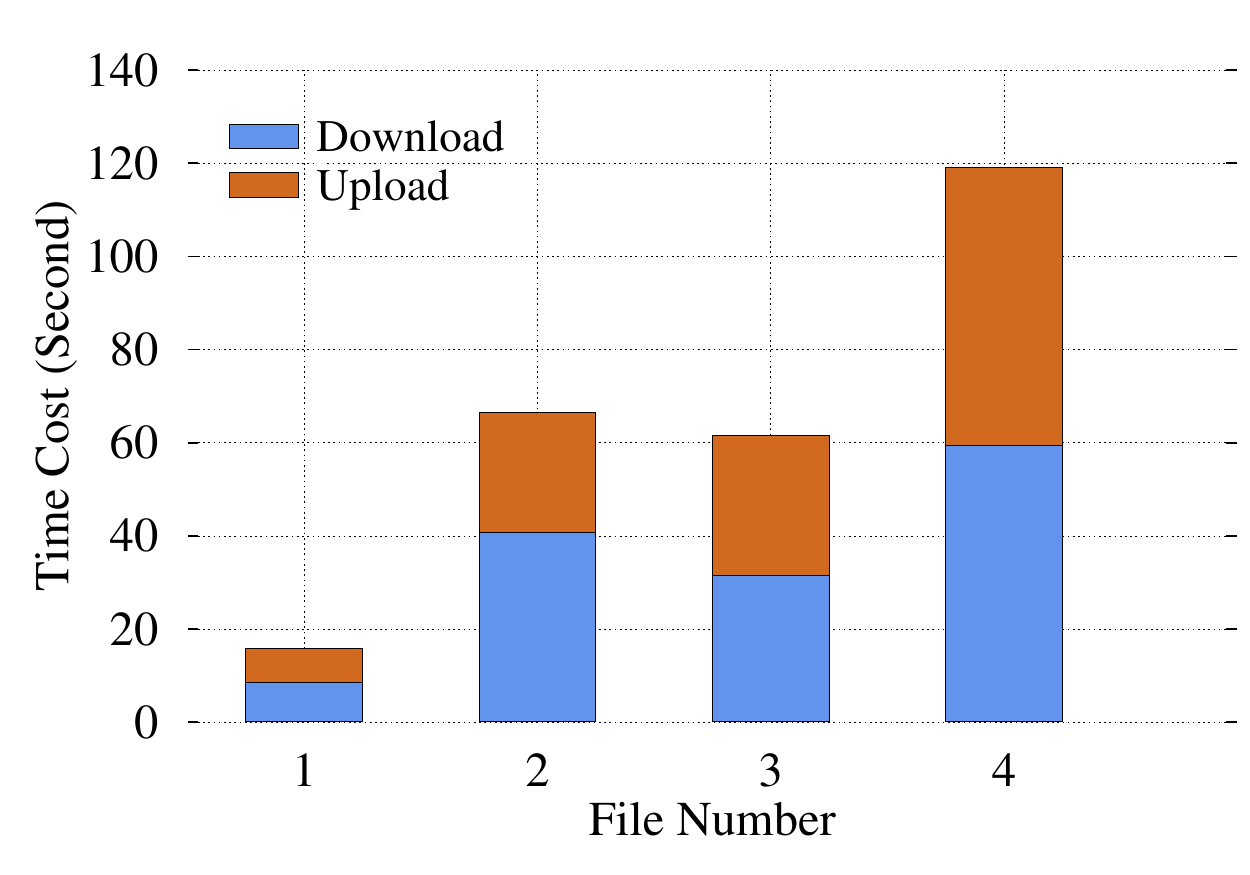}
\caption{Time cost of the file transfer (shared instances)}\label{share_time}
\end{minipage}
~
\begin{minipage}[t]{0.48\linewidth}
    \centering
\includegraphics[width=\linewidth]{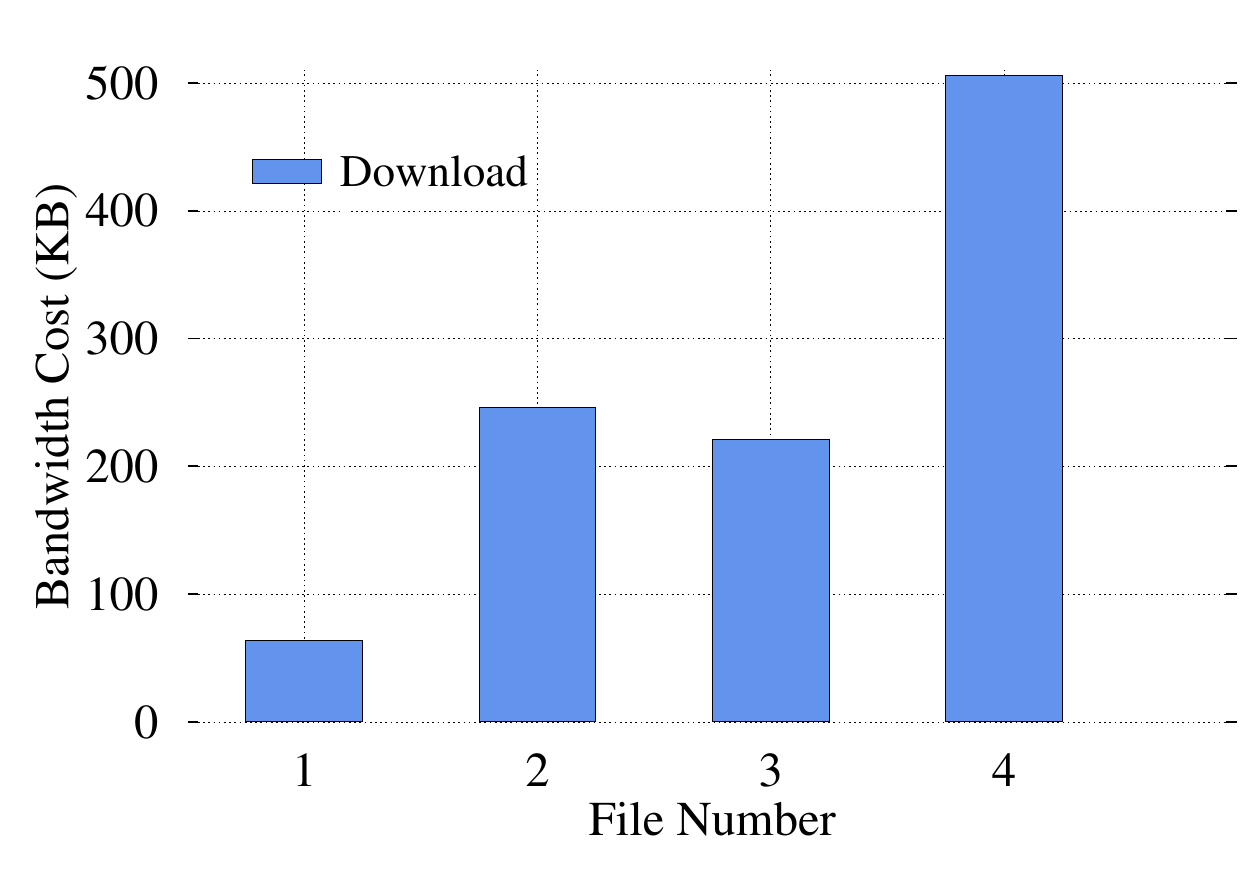}
\caption{Bandwidth cost of the file transfer (shared instances)}
\label{tab:file}
\end{minipage}
\vspace{-0.1in}
\end{figure}

\begin{table}[ht]
\centering
\begin{tabular}{|l|c|c|c|}
\hline
Workloads&Skyfiles&Traditional approach\\
\hline\centering
1&0.14\%&0.78\%\\
\hline
2&0.18\%&1.85\%\\
\hline
3&0.21\%&1.52\%\\
\hline
4&0.32\%&2.48\%\\
\hline
\end{tabular}
\caption{Power consumption of transferring files}
\label{tab:power}
\end{table}

Fig.~\ref{share_time} plots the time overhead of file transfer (both sender and receiver) with a shared
instance and Fig~\ref{tab:file} illustrates the bandwidth cost of the process on the sender. The binary service program hosted by the server is 4.92MB in our implementation. In the traditional approach, if the sender first fetches the file
to his smartphone and then transfers it to the receiver's smartphone, the
bandwidth consumption for both of them will be the same as the file size.
The time overhead will depend on not only the cellular network link quality
but also the smartphone-to-smartphone transfer protocol, e.g., NFC,
Bluetooth, or WiFi-Direct. The entire process is usually much slower relative to
Skyfiles's performance shown in Fig.~\ref{share_time}.

{\bf Power consumption:} The battery life is an important metric that directly affect
the user experience. We measure the Battery Percentage Usage (BPU) for transferring files operation on senders through 
a Nexus 5 smartphone.
The table~\ref{tab:power} shows the cost of BPU of Skyfiles and traditional approach under the LTE network.
As we can see from the table, Skyfiles outperforms traditional approach in all the tests. This is because that, in Skyfiles,
the actual file transfer is done on the cloud and the smartphones only need to exchange a small amount of control messages.

\section{Conclusion}
\label{sec:conclude}
In this paper, we develop Skyfiles, a cloud file management system, to help smartphone users execute 
operations on the files stored in cloud. 
The major objective of this system is to extend the boundary of operations so that current service providers support to a richer
operation set efficiently and securely. Skyfiles consists three components, Skyfiles Agent
Skyfiles Service Program and Skyfiles Server. On the mobile device side, Skyfiles Agent maintains a {\em shadow file system} that 
does not keep local copies and only stores the meta data, including file name, type and last modified time of the files.
On the cloud side, Skyfiles Service Program that residences in the cloud instance is used to provide various 
extended operations. Moreover, a centralized trusted server, Skyfiles Server, maintains the available instances. 
The system categorizes the user's request into basic and advance operations.
The basic operations, such as creating/deleting/renaming a file, are completed through the standard APIs defined by service providers.
The advance operations, like encryption and conversion, are achieved by Skyfiles Service Program.
To reduce the overhead and the cost of cloud instance, Skyfiles allows users to share the cloud instance with each other.
We propose the secure protocols to protect Skyfiles 
from disclosing users' security credential of their cloud storages when using shared cloud instance.
Skyfiles system implementation is evaluated on Android smartphones with Dropbox as a cloud storage platform and with
Amazon Web Service as a cloud instance provider.

\bibliographystyle{IEEEtran}
\bibliography{ref1}

\end{document}